\documentclass [12pt]{article}

\usepackage{amsmath,amsthm,amsfonts,amscd,latexsym,amssymb,amsbsy,dsfont,mathrsfs,color}
\usepackage[small, centerlast, it]{caption}
\usepackage{epsfig}
\usepackage[titles]{tocloft}
\usepackage[latin1]{inputenc}
\usepackage{verbatim} %for comments via \begin{comment}
\usepackage[active]{srcltx} %inverse search
\usepackage{fancyhdr}
\usepackage{caption}
\usepackage{enumerate}
\def\b1{{1\!\!1}}

\def\cB{\mathscr{B}}

\def\cI{{\ca I}}

\def\cL{\mathscr{L}}
\def\cM{\mathscr{M}}

\def\cS{\mathscr{S}}

\def\sM{{\mathsf M}}

\def\sH{{\mathsf H}}

\def\bC{{\mathbb C}}           %%%  complex numbers and so on

\def\bH{{\mathbb H}}

\def\bD{{\mathbb D}}
\def\bN{{\mathbb N}}

\def\bR{{\mathbb R}}

\def\gB{{\mathfrak B}}

\def\beq{\begin{eqnarray}}
\def\eeq{\end{eqnarray}}

\newcommand{\ca}[1]{{\cal #1}}         %%  calligraphic

\usepackage{sectsty}
\sectionfont{\normalsize}%\normalfont
\subsectionfont{\normalsize\normalfont\itshape}
\newtheoremstyle{plain}
{5pt}% space above
{9pt}% space below
{\itshape}% body font
{}% h indent amount
{\itshape\bfseries}% theorem head font
{}% punctuation after theorem head
{1em}% space after theorem head
{}% theorem head spec (can be left empty, meaning `normal')

\theoremstyle{thm}
\newtheorem{theorem}{\em Theorem}[section]
\newtheorem{lemma}[theorem]{\em Lemma}
\newtheorem{corollary}[theorem]{\em Corollary}
\newtheorem{proposition}[theorem]{\em Proposition}
\newtheorem{definition}[theorem]{\em Definition}
\newtheorem{remark}[theorem]{\em Remark}

\begin{document}

%%%%%%%%%%%%%   Title %%%%%%%%%%%%%%%%%%%%%%%%%%

\par
\bigskip
\large
\noindent
{\bf  The correct formulation of Gleason's theorem in quaternionic  Hilbert spaces}
\bigskip
\par
\rm
\normalsize

%%%%%%%%%%%%%%%%%%%%%%%%%%%%%%%%%%%%%%%%%%%%%
%%%%%%%%%%%% Authors %%%%%%%%%%%%%%%%%%%%%%%%

\noindent  {\bf Valter Moretti$^{a}$}, {\bf Marco Oppio$^{b}$}\\
\par

\noindent 
 $^a$ Department of  Mathematics University of Trento, and INFN-TIFPA \\
 via Sommarive 15, I-38123  Povo (Trento), Italy.\\
valter.moretti@unitn.it\\

\noindent   $^b$ Faculty of Mathematics,
University of Regensburg\\
Universit\"atsstrasse 31,
93053 Regensburg, Germany \\
marco.oppio@ur.de\\

 \normalsize

\par

\rm\normalsize

\noindent {\small  September, 6  2018}

%\linespread{1.5}
\rm\normalsize

%%%%%%%%%%%% Date %%%%%%%%%%%%%%%%%%%%%%%%%%

\par
\bigskip

\noindent
\small
{\bf Abstract}. 
Quantum Theories can be formulated in real, complex or quaternionic Hilbert spaces  as established in Sol\'er's theorem.  Quantum states are here pictured in terms of $\sigma$-additive probability measures over the  non-Boolean  lattice of orthogonal projectors of the considered Hilbert space.   Gleason's theorem proves that, if the  Hilbert space is either real or complex and some technical 
hypothes are true,  then these measures are  one-to-one with standard density matrices used by physicists recovering and motivating  the familiar notion of state.  The extension of this result to quaternionic Hilbert spaces was obtained by Varadarajan in 1968. Unfortunately,   the formulation of this extension \cite{V2a} is partially mathematically incorrect due to some peculiarities of the notion of trace  in quaternionic Hilbert spaces. A minor issue also affects Varadarajan's statement for real Hilbert space formulation. This  paper is devoted to present Gleason-Varadarajan's theorem into a technically correct and physically meaningful form valid for the three types of Hilbert spaces. In particular, we prove that only the {\em real part} of the trace enters the formalism of quantum theories (also dealing with unbounded observables and symmetries) and it can be safely used to formulate and prove a common statement of Gleason's theorem.  

\normalsize
\newpage

\tableofcontents

\section{Introduction}\label{sec1} 
\subsection{Gleason's theorem and troubles with  the quaternionic formulations}\label{sec1.1}
The idea to formulate Quantum Theories starting from the  partially ordered set  of {\em elementary propositions} can be traced back to  Birkhoff and von Neumann  \cite{BiNe}  with crucial contributions by Mackay \cite{Mackey} (see \cite{librone} for a review  on these issues).
From a physical point of view, these elementary propositions are the statements which can be physically tested on a quantum  system,  receiving either  the outcome $0$ (not valid) or the outcome $1$ (valid). In particular, every physical quantity  $A$ which can be measured on the system can be pictured as a {\em collection} of such elementary propositions $P^{(A)}_E$ labelled by sets of reals $E$ like an interval $E=(a,b)$. The physical meaning of $P^{(A)}_E$ is that {\em the outcome  of the observable $A$ belongs to $E$}.
The elementary propositions of a quantum system  obey a logic different from the classical one in view of the presence of physically {\em incompatible} statements (in the sense of Heisenberg principle). Von Neumann's idea was that these elementary propositions are one-to-one described by the orthogonal projectors of a complex Hilbert space $\sH$ associated to the quantum system
 and a pair of such elementary propositions are physically incompatible if and only if the corresponding projectors do not commute.  The set $\cL(\sH)$ of orthogonal projectors over $\sH$ enjoys the structure of a non-Boolean  {\em  lattice} with respect to the ordering relation given by the set-theoretic inclusion of projection subspaces. 
In particular,  the collection of elementary propositions associated to a physical quantity $A$ as above are supposed to be pairwise compatible and the analysis of the mutual relations between these propositions (see, e.g., \cite{V2a,M}) shows that they  give rise to  the structure of a {\em projection-valued measure} over $\bR$. At this point, the spectral machinery, integrating this projection-valued measure, permits one  to associate the abstract observable $A$  to a corresponding selfadjoint operator, motivating from a deeper viewpoint the standard assumption by physicists that observables are described by selfadjoint operators. 
It is also possible to assume an even more abstract viewpoint, where the elementary propositions are treated as elements of an abstract lattice with specific  mutual properties reflecting the general phenomenology of quantum systems, without explicitly referring to the orthogonal projectors (the notion of {\em commutativity} of pairs of elements can be generally defined for  abstract lattices exploiting  the notion of Boolean sublattice). Adopting this quite abstract point of view, it turns out that, in addition to the lattice of orthogonal projectors over a {\em complex} Hilbert space, also  the lattices of orthogonal projectors over either a {\em real} and a {\em quaternionic} Hilbert space fulfill the requirements which, in principle,  may be  justified by  the quantum phenomenology (irreduciblity, orthomodularity, $\sigma$-completeness, separablility, atomicity and validity of the so-called  covering property, see, e.g., \cite{BeCa}).
The proof of the fact that these are the {\em only three} possibilities has a long history (see, e.g., \cite{BeCa,XY, librone, M} for technical discussions of the various intermediate results), starting from several remarkable results by Piron and other few authors like F. Maeda, S. Maeda and  Araki in  the '60s,  and ending with the theorem by Sol\'er in 1995 \cite{Soler}. Assuming that the abstract lattice of quantum elementary propositions  is irreducible, orthomodular, $\sigma$-complete, separable, atomic, that it satisfies the covering property   and that it includes an infinite set of orthogonal atoms, then  the lattice is  necessarily  isomorphic to the  lattice $\cL(\sH)$  of orthogonal projectors in a separable Hilbert space  $\sH$.  There, the set of scalars may only be $\bR,\bC$ or the real algebra of quaternions $\bH$. 
 
We shall not address  here the problem of the apparent absence of physical systems described in real Hilbert spaces \cite{S1,S2,MO1} and the possibility (or impossibility) of quaternionic formulations   \cite{foudationofquaternionicmechanics,Adler,Gan17,MO2}. Instead, we restrict attention to a celebrated result, provided by Gleason's theorem \cite{G}, regarding the notion of {\em quantum state}.
The lattice structure of  $\cL(\sH)$ suggests a physically natural definition of a quantum state  (at a given instant of time $t$) defined as a map associating every elementary proposition $P\in \cL(\sH)$ with the probability that it results to be true if measured.   It should be evident that 
this is a good notion of state, as it includes all information necessary to compute probabilities, expectation values, standard deviations and so on for every observable (selfadjoint operator) viewed as the family of its spectral projectors (elementary propositions).
Formally, a state is  a generalized $\sigma$-{\em additive probability measure} over $\cL(\sH)$, that is a map  
$\mu : \cL(\sH) \to [0,1]$ satisfying both  $\mu(I) = 1$, and
 $\sum_{k\in K} \mu(P_k) = \mu\left( \vee_{k\in K} P_k\right)$ for every set $\{P_k\}_{k\in K} \subset \cL(\sH)$ with $K$ finite or countably infinite, such that $P_k$ and $P_h$ are orthogonal when $k\neq h$.
That this definition is in agreement with the standard notion of state familiar to  physicists was established  by Gleason with his celebrated theorem.

\begin{theorem}[{\bf Gleason's theorem}]\label{teoG}  
	Let $\sH$ be a  real or complex separable Hilbert space with dimension greater than $2$. The class of  $\sigma$-additive probability measures  $\mu$ over $\cL(\sH)$ is one-to-one  with the class of  self-adjoint, positive, trace-class operators $T$ with unit trace. This correspondence  is defined by the requirement  
	\begin{equation}\label{trace}
	 \mu(P) = tr(PT)\:,
	 \end{equation}
	for every choice of $P \in \cL(\sH)$.
\end{theorem}
 This way,  the notion of {\em statistical operator} $T$ (self-adjoint, positive, trace-class operators with unit trace)
enters the mathematical formulation of Quantum Theories. 
However, also states  represented by normalized vectors of $\sH$,  show up.
To this end, first observe that a {\em convex} combination $p\mu+q\mu'$ of probability measures  (i.e., of statistical operators $pT+qT'$) -- where $p+q=1$ and $p,q\in [0,1]$ by definition of convex combination --  is still a probability measure (a statistical operator).
 Among the class probability measures over $\cL(\sH)$ there are elements, defining the  so-called {\em pure states}, which cannot be decomposed into  non-trivial convex combinations of  probability measures. They are by definition the {\em extremal elements} of the convex body of the said measures. It is easy to prove that they are the usual vector states  of the standard formulation familiar to physicists, i.e.,  the associated statistical operators have  the form 
$T = \psi\langle \psi| \cdot\rangle$, where $||\psi||=1$. The spectral theorem finally proves that a generic statistical operator is always  a (generally infinite)  convex combinations of pure states, i.e., a {\em quantum mixture} or an {\em incoherent superposition} of pure states  familiar to physicists.

\begin{remark}$\null$
{\em \begin{itemize}
\item[(a)] It is worth stressing that pure states 
are unit vectors {\em up to signs} in the real formulation, whereas they
are unit vectors {\em up to phases} in the complex formulation. So that, in particular, real quantum mechanics does not coincide with  {\em decomplexified} complex quantum mechanics where pure states would be unit vectors up to $SO(2)$ rotations.
\item[(b)] 
The statement of Gleason's theorem is also trivially true for $\dim(\sH)=1$, whereas the constraint $\dim(\sH)\neq 2$ cannot be removed, since well-known counterexamples can be constructed (see, e.g., \cite{M}).  There exist several extensions of Gleason's result towards various directions, e.g., dealing with a lattice of orthogonal projectors of a von Neumann algebra in $\sH$ instead of the whole $\cL(\sH)$, or relaxing $\sigma$-additivity, or positivity  of $\mu$ or separability of $\sH$. Remarkably,  the requirement $\dim(\sH)\neq 2$ or a corresponding constraint in terms of von Neuman algebras of definite type (type-$I_2$ is forbidden) survives all extensions.
 An exhaustive survey on the subject is \cite{libroGleason}. In the rest of the paper,  we stick to the elementary version represented by the statement given above of Gleason's theorem. 
\item[(c)] 
The most important  physical consequence of Gleason theorem is proving that the idea of states viewed from skratch  as probability measures over the lattice of elementary is in  agreement with the picture, more  familiar to physicists,  where   the building blocks are pure (vector) states and statistical operators are secondary objects. 
Failure of Gleason's theorem would imply that there are two inequivalent notions of state in the elementary  formulation of Quantum Theories in Hilbert space.
\item[(d)] 
In physically important situations where not all the selfadjoint operators represent observables, as  in the presence of superselection rules or gauge groups, the idea of states as probability measures over the {\em restricted} sublattice of physically meaningful projectors reveals to be more useful than the standard approach based on vectors states and their quantum mixtures. As is known, in the presence of superselction rules, {\em many} different state vectors and  statistical operators may contain the same information. Conversely there is {\em exaclty  one} probability measure over the projectors representing the elementary propositions permitted by the superselection rules  associated to a class of equivalent statistical operators and vector states. It is possible to extend  Gleson's theorem to encompass these cases  in order to classify all possible statistical operators associated to a given probability measure when, for instance, superselction rules affect the theory (e.g., see \cite{M}).
\item[(e)]  Within the framework based on the notion of lattice $\cL(\sH)$ of elementary propositions and 
 on the idea of quantum states  viewed as  probability measurs, the notion of trace plays a crucial tole. It  is in fact  a general mathematical tool useful to charcterize the structure of the probabilty measures over $\cL(\sH)$ somehow extending  the notion of integral over non Boolean algebras. However, it has not an {\em a priori}  direct physical meaning which would have if reversing the construction as in the  formulation more familar to physicists, where the fundamental objects are pure states and the operation of trace is often introduced with physical motivations to describe loss of coherence or (partial trace) loss of information on subsystems.  
\end{itemize}}
\end{remark}

All the discussion above concerned the real and the complex Hilbert space cases. Let us pass to focus on the case of a quaternionic Hilbert space which is also permitted as a model of elementary propositions. Quantum states can be defined as probability measures over $\cL(\sH)$  as well. Are these probability measures one-to-one represented by statistical operators?
Differently from what erroneously asserted in \cite{V2a,V2}, the formulation of Gleason's result cannot directly encompass the quaternionic Hilbert space case in the form stated in Theorem \ref{teoG}. This is because,
\begin{enumerate}[{\bf (A)}]
\item in quaternionic Hilbert spaces the notion of trace is generally {\em basis dependent} as noticed in \cite{T} and  \cite{CGJ16}, unless its argument is selfadjoint, so that Eq.(\ref{trace}), where $PT$ is generally {\em not} selfadjoint, needs further specification. (As a trivial example of basis dependence, using the notation of  Sect.\ref{secnot} below, consider the algebra of quaternions as one-dimensional quaternionic Hilbert space with scalar product 
$\langle q|q' \rangle := \overline{q}q'$. Both $1$ and the first imaginary units $i$ are Hilbert bases. The operator $Jq := jq$ for $q \in \bH$ have different traces: $\langle 1|J1\rangle = j$ and $\langle i|Ji\rangle = -j$.)

\item  $tr(PT)$ in Eq.(\ref{trace}) would  produce quaternions rather than reals in $[0,1]$ {\em as it should be since $tr(PT)$ has the meaning of a probability}. This is because the {\em cyclic property} of the trace $tr(PT)=tr(TP)$ generally {\em fails} for quaternionic Hilbert spaces so that selfadjointness of $T$ and $P$ is not enough to guarantee $tr(PT)\in \bR$ (whereas  it is sufficient  in the complex case: $\overline{tr(PT)} = tr(T^*P^*) =tr(TP) = tr(PT)$).
 
\item A minor issue, already noticed in \cite{MO1},  also affects the real-Hilbert space case in Varadarjan's formulation, since the very definition of trace-class (bounded) operator  $A: \sH \to \sH$ adopted in \cite{V2} for the three types of Hilbert spaces:
\beq \sum_{x\in N} |\langle x|A x\rangle| < +\infty\quad \mbox{for every Hilbert basis $N \subset \sH$}\label{condwrong}\eeq
 is ineffective in the real case. In  infinite-dimensional real Hilbert spaces, there are nonvanishing bounded operators $A: \sH \to \sH$ such that $\langle z|A z\rangle =0$ for every $z\in \sH$ so that (\ref{condwrong}) holds, but  it is {\em false} that 
 \begin{equation}\label{condwrong2}
 \sum_{x\in N} \langle x||A| x\rangle< +\infty\quad \mbox{for every Hilbert basis $N \subset \sH$}\:.
 \end{equation}
  (It is sufficient to construct $A$ such that $A^*=-A$ and $AA=-I$ which entails that $|A|=I$ so that the left-hand side of (\ref{condwrong}) diverges for every choice of $N$.) Even if (\ref{condwrong}) implies that some notion of trace can be defined, we cannot drop requirement (\ref{condwrong2}) since it guarantees that the standard properties of trace-class operators are valid as we shall discuss in Proposition \ref{propTRACE1}.
\end{enumerate}

 If $\sH$ is complex or quaternionic,  the counter-example in (C) does not work because (see Exercise 3.21 in \cite{M} for $\bD=\bC$ and Proposition 2.17 (a) \cite{GMP1} for $\bD=\bH$)  $\langle z|A z\rangle =0$ for every $z\in \sH$ imply $A=0$ and the problem disappears. In real Hilbert spaces, this result is achieved only if $A=A^*$, it easily follows  from
$2\langle x|Ay \rangle +  2\langle y|Ax \rangle= \langle x+y|A(x+y) \rangle - \langle x-y| A(x-y) \rangle$.
In complex and quaternionic Hilbert spaces actually (\ref{condwrong2}) and (\ref{condwrong}) are equivalent, as we shall discuss shortly.

In spite of these technical difficulties, the statement in \cite{V2a,V2} seems  physically safe since  everything can be re-arranged in physical applications \cite{V2a,V2} to obtain the expected results. Physical soundness  
 together with the absence of a sufficiently advanced rigorous formulation of  quaternionic Hilbert space  functional analysis  when the book  was written (1968)
are, in our view,  the reasons why these problems did not produce consequences into the physical literature. Furthermore,  the central and hard  part of Varadarajan's proof  is correct, and the final statements can be fixed into a formulation of  Gleason's theorem, presented as Theorem \ref{GVt},  that is valid for the three kinds of Hilbert spaces. This is the main goal of this paper. The key tool to formulate this common version of the  theorem is the notion of {\em real trace} that will be introduced shortly. After the formulation of the theorem we shall show how this notion is in complete agreement with the standard formalism used by physicists.

\subsection{Known results and notations}\label{secnot}
The real associative division algebra of {\bf quaternions}  is   $\bH =\{a+ bi + cj + dk\:|\: a,b,c,d \in \bR\}$. The three imaginary unitis $i,j,k$ pairwise anticommute, satisfy the well-known relations  $i^2=j^2=k^2=-1$,
and $ij = k$ and cyclic permutations of it. The {\bf quaternionic  
conjugation} is defined as $\overline{a+ bi + cj + dk}= a- bi - cj - dk$ for $a,b,c,d \in \bR$. Notice that $\overline{pq}= \overline{q}\:\overline{p}$ if $p,q \in \bH$. The {\bf real part} of a quaternion is therefore $Re(q) = \frac{1}{2}(q+ \overline{q})$.  
Notice that $Re(qq') = Re(q'q)$ for $q,q'\in \bH$ as it immediately arises from a direct computation.
$\bH$ is a real Banach  space if referring to the natural {\bf norm} $|a+ bi + cj + dk|:= \sqrt{a^2+b^2+c^2+d^2}$ that also satisfies $|pq|=|p|\:|q|$ and $|q| = \sqrt{\overline{q}\:q} = |\overline{q}|$ if $p,q\in \bH$. In particular $\bH$ is a {\em unital real $C^*$-algebra}, the $^*$ operation being the quaternionic conjugation.

 We assume that the reader is familiar with standard notions of Hilbert space theory (including linear operators thereon) both for the real  (see, e.g., \cite{Brezis}) and complex case (see, e.g., \cite{M}) and we pass to introduce the quaternionic case.

We henceforth use the notion of quaternionic vector space and its properties as defined in \cite{GMP1}. It is worth stressing that, 
in \cite{GMP1,GMP3,GMP2}, the multiplication of a vector of a quaternionic Hilbert space  $x\in \sH$ and a quaternion $q\in \bH$ was supposed to act {\em on the right},
$$\sH \times \bH \ni  (x,q) \mapsto xq \in \sH\:.$$
 In view of non-commutativity of quaternions, this choice is compulsory  as soon as one requires that the inner product of a Hilbert space is linear in the {\em right entry}, as generally  assumed in mathematical physics literature.
 \begin{remark}\label{remsom} {\em  Since we want to use a common notation for real, complex, and quaternionic Hilbert spaces, we therefore adopt a right multiplication of scalars and vectors instead of the standard left one, also in real and complex Hilbert spaces.  With this choice, linearity is written, if $\sH$ is a real, complex or quaternionic vector space, as
$$A(xq) = (Ax)q \:, \quad x \in \sH,\: q \in \mbox{$\bR$ of $\bC$ or $\bH$ respectively.}$$
Linear combination of operators are defined according to this convention. However, 
in case of quaternionic Hilbert spaces, linear combinations of operators are therefore permitted only if the coefficients of the combination are real, since only reals in $\bH$ commute with all quaternions. We stick to the standard notation $aA + bB$ to denote linear combination of (right-linear) operators $A, B : \sH \to \sH$ 
 $$(aA + bB)(x) := Axa + Bxb\:,$$
where $x \in \sH$ and $a,b \in \bR$ for both real or quaternionic Hilbert spaces, whereas $a,b \in \bC$ for complex Hilbert spaces.}
\end{remark}
 A {\bf scalar product} $\langle \cdot |\cdot \rangle : \sH \times \sH \to \bH $ over a quaternionic vector space $\sH$ is, by definition, $\bH$-linear in the right entry $\langle x| yq+ zp\rangle = \langle x| y\rangle q + \langle x| z\rangle  p$, positive $\langle x|x\rangle \geq 0$ with $x=0$ if  $\langle x|x\rangle =0$, and Hermitian $\langle x|y\rangle = \overline{\langle y|x\rangle}$, for every $x,y \in \sH$, $p,q \in \bH$.  \\$\sH$ is a {\bf Hilbert space} if  the {\bf norm} $||x|| := \sqrt{\langle x|x\rangle}$  makes  $\sH$  complete as a metric space.  

 We henceforth denote by $\bD$ the set of scalars of the Hilbert space $\sH$ which can be {\em real} ($\bD=\bR$), {\em complex} ($\bD=\bC$), or {\em quaternionic} ($\bD=\bH$).
 The scalar product is henceforth denoted by  $\langle \:\:|\:\:\rangle$ in the three types of Hilbert space.

The maps $ \bD \ni r  \mapsto Re(r) \in  \bR$ and $\bD \ni q  \mapsto \overline{q} \in \bD$ are therefore defined in the three cases with the standard meaning for $\bD=\bC$  and reducing to the identity map if $\bD=\bR$.

If $\bD=\bH$,
the notion of {\em Hilbert basis},  {\em orthogonal complement}, {\em bounded operator},  {\em adjoint  operator}, {\em operator norm}, {\em uniform}, {\em strong} and {\em weak operator topologies}  are  defined exactly as in real or complex Hilbert spaces and enjoy identical properties with the obvious trivial changes \cite{GMP1}. The (Hermitian) adjoint of an operator  $A$ is always denoted by $A^*$.

 A {\bf unitary} operator $U : \sH \to \sH$ is a norm-preserving surjective $\bD$-linear operator. In the special case of $\bD=\bC$, an {\bf anti-unitary} operator $U : \sH \to \sH$ is an anti-$\bC$-linear norm-preserving surjective operator. By polarization, unitary operators  preserve the scalar product, while anti-unitary operators do it up to an overall complex conjugation of the scalar product ($\langle Cx|Cy\rangle = \overline{\langle x|y\rangle}$). 

If $\bD=\bH$, the set $\gB(\sH)$ of bounded operators $T : \sH \to \sH$ has the natural structure of a {\em unital real $C^*$-algebra} (the $^*$-operation being the Hermitian adjoint), since only {\em real} linear combinations of ($\bH$-linear) operators are defined. Evidently $\gB(\sH)$ is also a unital real $C^*$-algebra if $\bD=\bR$, whereas $\gB(\sH)$  is a proper unital complex $C^*$-algebra when $\bD=\bC$.

 $\cL(\sH)$ denotes the set of {\bf orthogonal projectors} over $\sH$, i.e., of operators $P\in \gB(\sH)$ such that $PP=P$ and $P^*=P$. Orthogonal projectors $P$ are one-to-one with the class of all closed subspaces $\sM$ of $\sH$ through the requirement $P(\sH) = \sM$.
$\cL(\sH)$ is an  {\em orthomodular} lattice 
 with respect to the partial ordering relation $P\geq Q$ iff $P(\sH) \supset Q(\sH)$
and the orthocomplement $P^\perp =I-P$ ($P,Q \in \cL(\sH)$). $\cL(\sH)$ is {\em complete}: the {\em supremum} $\vee_{j\in J}P_j$ exists in  $\cL(\sH)$ for every $\{P_j\}_{j\in J} \subset \cL(\sH)$ (see, e.g.,  Theorems 7.22 and 7.56 in \cite{M}, the proofs being identical in the three cases).

$T\in \gB(\sH)$ is said to be {\bf positive} when $\langle u|Tu\rangle \geq 0$ if $u\in \sH$. 
The {\em polar decomposition} $A= U|A|$ of an operator $A \in \gB(\sH)$, its {\bf absolute value} $|A| := \sqrt{A^*A}$ and the {\em positive squared root} $\sqrt{B}$ of a 
selfadjoint positive operator $B \in \gB(\sH)$ are defined for quaternionic Hilbert spaces \cite{GMP1} exactly as for real (see, e.g.,\cite{MO1}) and complex Hilbert spaces (see, e.g., Sect.3.5.2 in  \cite{M}): $\sqrt{B}$ is the unique selfadjoint positive operator in $\gB(\sH)$ such that $\sqrt{B}\sqrt{B}=B$. In complex and quaternionic Hilbert spaces, positive operators are also 
selfadjoint, since positivity implies in particular that $\langle x|Ax \rangle= \langle Ax|x \rangle$ and so
$\langle x|(A^*-A)x \rangle =0$ for every $x\in \sH$ that, in turn, leads to $A^*=A$ (see the comment just below {\bf (C)} above). In real Hilbert spaces, positivity generally does not implies selfadjointness.  
If $S,T$ are bounded operator, we say that $T\ge S$ whenever $T-S\ge 0$. It should be noticed that in the real Hilbert space case $\ge$ is \textit{not} a partial ordering relation  in $\gB(\sH)$, the antisymmetry axiom not being satisfied (unless restricting to the set of selfadjoint opertors in $\gB(\sH)$).

The notion of {\em compact} quaternionic-linear operator is identical to that in real or complex Hilbert spaces since it uses the real structure only. 
The set of compact operators will be denoted\footnote{That set was denoted by $\gB_0(\sH)$ in \cite{GMP3}.} by $\gB_\infty(\sH)$. With the same proof as for complex Hilbert spaces (Proposition 4.14 and  Theorem 4.15 in \cite{M}), it is easy to prove the following proposition irrespective of the nature of $\bD$.
\begin{proposition}\label{prop1}
	 Let $\sH$ be a real, complex or quaternionic Hilbert space. $\gB_\infty(\sH)$ is a subspace and a {\em closed two-sided $^*$-ideal} of $\gB(\sH)$. Furthermore  $T \in \gB_\infty(\sH)$ if and only if $|T| \in \gB_\infty(\sH)$. 
\end{proposition}
The {\em Hilbert decomposition theorem over a Hilbert basis of eigenvectors}, for compact self-adjoint  operators (Theorem 4.20 in \cite{M} for $\bD=\bC$ which is also valid for $\bD=\bR$ as is easy to check) holds true also for $\bD=\bH$  as proved in Theorem 1.2 in \cite{GMP3}. 

We specialize the statement to the case of a selfadjoint compact operator:
\begin{proposition} \label{propdecT} 
	Let $T^*=T\in \gB_\infty(\sH)$ where $\sH$ is a real, complex or quaternionic Hilbert space (not necessarily separable). There is a Hilbert basis  $N$ of $Ker(T)^\perp$ made of eigenvectors of $T$ such that $N$ is finite or countably infinite and
	\begin{equation} 
	Tx = \sum_{u\in N}  u s(u)\langle u|x\rangle  \quad \forall x \in \sH \label{series1}\:,
	\end{equation} 
	$s(u)\in \bR$ is the eigenvalue of $u$ whose multiplicity (the dimension of the associated eigenspace of $T$)  is finite and $\{s(u)\}_{u\in N} = \sigma(T)\setminus\{0\} =  \sigma_p(T)\setminus\{0\}$
where, in case of $\bD= \bH$ the said spectra are intepreted as spherical spectra (see remark below). 
The series (\ref{series1})  can be re-ordered arbitrarily preserving the sum.
\end{proposition}
\begin{remark}\label{rempropdecT} {\em Some remarks on the $\bD=\bH$ case are  listed below.
	\begin{enumerate}[(a)]
		\item 
Since quaternions are not commutative,  the standard notion of spectrum must be changed in quaternionic Hilbert spaces.
As  first discussed in \cite{7,8} and later adapted to the Hilbert space theory in  \cite{GMP1} with a finer definition,  the relavant quaternionic notion of spectrum 
does not concern the properties of the operator $T-\lambda I$ when $\lambda$ varies in $\bR$ or $\bC$, but those of the quadratic operator $\Delta_q(T) := T^2-T(q+\overline{q}) + I|q|^2$ when $q$ varies in $\bH$. The {\em spherical spectrum} $\sigma_S(T)$ of $T\in \gB(\sH)$ is made of all quaternions $q\in \bH$ such  that the {\em spherical resolvent} $\Delta_q(T)^{-1}$ does not exist in $\gB(\sH)$. 
It is easy to see that, if
$q\in \sigma_S(T)$, then  $pqp^{-1}\in  \sigma_S(T)$ for every $p\in \sH$ with $|p|=1$.  This property can be described as the action of $SO(3)$ on the non-real part of $q$ interpreted as an element of $\bR^3$ and this is the reason why $\sigma_S(T)$ is called {\em spherical}.
 The generalization to unbounded and not everywhere defined  operators is straightforward \cite{GMP1} and the  finer decomposition into {\em point spherical spectrum}, {\em continuous spherical spectrum} and {\em residual sperical spectrum} is strictly analogous to the corresponding decomposition of the standard spectrum in real and complex Hilbert spaces. As a matter of fact,  it consists of sistematically replacing $(T-\lambda I)$ for $\Delta_q(T)$ in every standard definition in real and comples Hilbert spaces (see \cite{GMP1} for details). In particular, the {\em point spherical spectrum} $\sigma_{pS}(T)$ is made of the quaternions $q$ such that  $\Delta_q(T)$ is not injective. 
Though apparently unrelated with the analogous definition in real and complex Hilbert spaces,  the notion  of spherical spectrum is that exploited in the extensions of all results of spectral theory (e.g., see \cite{GMP1,ACK16,GMP2}). 
		\item In spite of the different definition of (spherical) spectrum of $T$, the elements $q$ of the {\em spherical point spectrum} $\sigma_{Sp}(T)$ turn out to be the {\em eigenvalues} of $T$ (Proposition 4.5 in \cite{GMP1}): there exists $u\in \sH\setminus \{0\}$ such that $Tu=uq$ and $u$ is called {\em eigenvector} of $q$ as usual. As a consequence, the  sperical {\em point} spectrum of an operator in a quaternionic Hilbert space coincides  with the set of eigenvectors of that operator  as it happens for the standard point spectrum in real or complex Hilbert spaces.
		\item Though this is generally false in view of the fact that quaternions does not commute, 
the eigenspaces of a {\em selfadjoint} operator $T$ are $\bH$-subspaces of $\sH$, since the eigenvalues are real  (Theorem 4.8 in \cite{GMP1}) and  quaternionic combinations of eigenvectors of a given (real) eigenvalue are still eigenvectors with the same eigenvalue. 
\item When $\bD=\bH$, the spectrum of the operator $T$ in Proposition \ref{propdecT}  has to be interpreted as its {\em spherical spectrum}  which is however a subset of $\bR$ since $T=T^*$. Furthemore,  possibly up to $0$ which however does not play any role in the considered decomposition of $T$,  the spherical spectrum of $T$ in Proposition \ref{prop1}   coincides with the the set of eigenvalues of $T$, in accordance with the remark (b) above.
	\end{enumerate} }
\end{remark}

The {\em spectral theorem} for (densely defined, closed, and generally unbounded) normal operators and in particular  selfadjoint operators can be formulated in complex Hilbert spaces (see e.g., \cite{R,S,M}) and quaternionic Hilbert spaces \cite{ACK16,GMP2}. For the real case a corresponding theorem is suitable for selfadjoint (generally unbounded) operators (e.g., see \cite{Li,MO1}). 
 \begin{remark}\label{remsom2} {\em 
The spectral machinerey developed in \cite{GMP2} permits one to integrate {\em quaternionic-valued} functions with respect to a projector-valued measure (PVM). The problem pointed out in Remark \ref{remsom} (apparently,  only $\bR$-linear combination of projectors are permitted in quaternionic Hilbert spaces)
  is not an obstruction. A PVM in a quaternionic Hilbert space  is in fact ``intertwining'' \cite{GMP2}, i.e.,    equipped with a so-called {\em left multiplications}  $L: \bH \ni q \mapsto L_q \in \gB(\sH)$ commuting with the PVM itself and giving rise to an operatorial representation of $\bH$. $L$ permits to define {\em quaternionic-linear} combinations of projectors  of the associated PVM as $L_qP_E+ L_{q'}P_F$ (coinciding to $qP_E+ q'P_F$ for $q,q'\in \bR$). These $\bH$-linear combinations of projectors  are  the  building-blocks used to define the spectral integral  of quaternionic-valued functions within the precise formulation of the quaternionic spectral theory established in  \cite{GMP2}. This work uses integrals of real-valued functions only,  so that  the left-multiplication of a PVM does not affect any result (though it exists). However, the full machinery of quantum mechanics exploits the general integration procedure, for instance dealing with continuous representations of Lie groups of symmetries \cite{MO2}.}
\end{remark}

\subsection{Structure of the work}
The next section is devoted to extend the basic theory of trace-class operators to the case of quaternionic (generally non-separable) Hilbert spaces, focussing in particular on the notion of {\em real trace} that naturally arises when requiring basis independence.
The subsequent section proposes a common formulation of Gleason's theorem relying to the correct part of Varadarajan formulation and replacing the notion of trace with that of real trace. The last section studies how our general formulation of Gleason's theorem and the notion of real trace is compatible with the standard formalism  handled by physicists. We find a total agreement.
A final appendix includes the proofs of some technical propositions spread throughout the work.

\section{Notions of trace}
We introduce here a common notion of trace-class operators including real, complex and quaternionic Hilbert spaces and  corresponding notions of traces.
\subsection{Trace class operators}
The notion of trace-class operator we introduce here is the direct extension to $\bR$ and $\bH$ of the standard notion adopted in complex Hilbert spaces (e.g., see Sect.4.4 of \cite{M}).
For the quaternionic case we make use of the result established  in \cite{GMP3} in particular Proposition \ref{propdecT} and the remark below it.
\begin{definition}\label{defdefT} 
	Let $\sH$ be a real, complex or quaternionic Hilbert space. An operator $T\in \gB(\sH)$ is said to be of {\bf trace class} if 
	\begin{equation}\label{defT}
	\sum_{u \in N} \langle u||T|u\rangle< +\infty\:,
	\end{equation}
	for some Hilbert basis $N\subset \sH$.  $\gB_1(\sH)\subset \gB(\sH)$ denotes the set of trace-class operators.
\end{definition}

\noindent Some of the relevant properties of trace-class operators are summarized below.
 We explicitly  omit to discuss the interplay of trace-class operators and {\em Hilbert-Schmidt} ones for the sake of simplicity. Furthermore, no reference to the theory of {\em Schatten class} operators will be mentioned. We address the reader to \cite{CGJ16} for a recent extension of the theory to  quaternionic Hilbert spaces, where a different notion of trace class
operators  and a full definition of  Schatten class of operators   in quaternionic Hilbert spaces are proposed referring to a preferred, arbitrarily fixed, anti-selfadjoint unitary operator. No choice of such an operator is made within our approach. 

\begin{proposition}\label{propTRACE1}  
	Let $\sH$ be a real, complex or quaternionic Hilbert space. The set $\gB_1(\sH)$ of trace-class operators enjoys the following properties.
\begin{enumerate}[(a)]	
	\item If $T \in \gB_1(\sH)$, then (\ref{defT}) is valid for every  Hilbert basis $M\subset \sH$ and  $\sum_{u \in M} \langle u||T|u\rangle$ does not depend on  $M$.
	\item $T \in \gB_1(\sH)$ if and only if both
	\begin{enumerate}[(i)]
		\item $T \in \gB_\infty(\sH)$ and 
		\item the following fact is true
		$$
		||T||_1 := \sum_{s \in \sigma_p(|T|)} s d_s < +\infty\:,
		$$
		where $\sigma_p(|T|) \subset [0, +\infty)$ is the point-spectrum of the selfadjoint compact operator $|T| := \sqrt{T^*T}$ (for $\bD=\bH$, $\sigma_p(|T|)$ is the spherical point-spectrum of $|T|$), and $d_s=1,2,\ldots <+\infty$ is the dimension of the $s$-eigenspace of $|T|$. 
	\end{enumerate}
	\item If $T\in\gB_1(\sH)$, then for every Hilbert basis $M \subset \sH$ it holds that
	$$
	||T||_1 = \sum_{u \in M} \langle u||T|u\rangle
	$$
	\item The set $\gB_1(\sH)$ has the structure of
	\begin{enumerate}[(i)]
		\item  $\bR$-linear subspace of $\gB(\sH)$  and also $\bC$-linear if $\bD=\bC$,
		\item two-sided $^*$-ideal of $\gB(\sH)$,
		\item real Banach space with respect to the norm $||\cdot||_1$.
	\end{enumerate}
	\item The following facts are true for $A \in \gB_1(\sH)$ and $B\in \gB(\sH)$.
	\begin{enumerate}[(i)]
		\item $||AB||_1\leq ||A||_1||B||$  and $||BA||_1\leq ||A||_1||B||$,
		\item $||A||_1= ||A^*||_1$,
		\item $||A|| \leq ||A||_1$.
	\end{enumerate}
\end{enumerate}
The set $\gB_1(\sH)$ is therefore a real (complex if $\bD=\bC)$  Banach algebra which is also a  $^*$-algebra and the norm is $^*$-invariant.
\end{proposition}

\begin{proof} See Appendix \ref{Appendix}. \end{proof}

\subsection{Real trace}
 We can now pass to discuss the notion of trace stating and proving some  properties relevant for our final goal. We also include some of the results established in  \cite{T} (here extended also to the non-separable case).
\begin{proposition}\label{propTRACE}  
	Let $\sH$ be a  Hilbert space over $\bD= \bR$, $\bC$, or $\bH$ not necessarily separable. If   $N \subset \sH$ is a Hilbert basis and $T\in \gB_1(\sH)$, then the {\bf $N$-trace} of $T$, 
	\begin{equation}\label{eqtr}
		tr_N(T) := \sum_{x\in N} \langle x|T x\rangle
	\end{equation}
	is a well defined element of $\bD$ satisfying 
	\begin{equation}\label{ineqtr}
		 |tr_N(T)| \leq ||T||_1\:.
	\end{equation}
	The right-hand side of (\ref{eqtr}) consists of a sum over a set at most   countable of non-vanishing elements. That series absolutely converges and, for the given $N$, can be re-ordered arbitrarily preserving the sum. 

	The further facts are true.
	\begin{enumerate}[(a)]
		\item If  $a,b \in \bR$ (or $\bC$ if $\bD=\bC$) and $A,B \in \gB_1(\sH)$,
		\begin{enumerate}[(i)]
			\item $tr_N(A^*)= \overline{tr_N(A)}$,
			\item $tr_N(aA+bB) = a tr_N(A) + b tr_N(B)$,
			\item $tr_N(A)\ge 0$ whenever $A\ge 0$.
		\end{enumerate}
		\item If $\bD=\bR$ or $\bC$ and $A\in \gB_1(\sH)$, $B \in \gB(\sH)$ then
		\begin{enumerate}[(i)]
			\item $tr_N(A)= tr_{N'}(A)$,
			for every pair of Hilbert basis $N,N'\subset \sH$,
			\item $tr_N(AB) = tr_N(BA)$.
		\end{enumerate}
		\item If $\bD=\bH$ and $A\in  \gB_1(\sH)$, then the following facts are equivalent.
		\begin{enumerate}[(i)]
			\item $tr_N(A)= tr_{N'}(A)$ for every pair of Hilbert basis $N,N'\subset \sH$.
			\item $A=A^*$.
		\end{enumerate}
			\item If $A\in  \gB_1(\sH)$, $B \in \gB(\sH)$, the followig facts hold,
		\begin{enumerate}[(i)]
			\item $Re(tr_N(A))= Re(tr_{N'}(A))$ for every pair of Hilbert basis $N,N'\subset \sH$,
			\item $Re(tr_N(AB)) = Re(tr_N(BA))$.
		\end{enumerate}
		\item If $A=A^*\in\gB_1(\sH)$ and $N_A$ is a Hilbert basis of $\sH$ obtained by  adding a Hilbert basis of $Ker(A)$ to a Hilbert basis of $Ker(A)^\perp$ made of eigenvectors of $A$ as in Proposition \ref{propdecT}, then
		\begin{equation}
		tr_{N_A}(BA)=tr_{N_A}(AB)
		\end{equation}
		for any $B\in\gB_1(\sH)$. Furthermore, if $B=B^*$, then $tr_N(AB)\in\bR$.
		\item If $A,B\in\gB_1(\sH)$ satisfy $A\ge B$, then 
		$$
		Re(tr_N(A))\ge Re(tr_N(B))
		$$
		for every Hilbert basis $N\subset \sH$.
	\end{enumerate}
\end{proposition}

\begin{remark}
{\em Notice that, although $tr$ does not satisfy the cyclic property  $tr(AB)=tr(BA)$ in quaternionic Hilbert spaces, its real part does as stated in (d)(ii).}
\end{remark}
\begin{proof}
See Appendix \ref{Appendix}.
\end{proof}

Proposition \ref{propTRACE} leads to the following definition.
\begin{definition}\label{DEFTRV} 
	Let $\sH$ be a real, complex or quaternionic Hilbert space and $A \in \gB_1(\sH)$.
	\begin{enumerate}[(a)]
		\item The {\bf real trace} of $A$ is defined as
		$$
		tr^{\bR}(A) := Re(tr_N(A))
		$$ 
		where $N\subset \sH$ is a Hilbert basis.
		\item If $tr_N(A)$ does not depend on the Hilbert basis $N \subset \sH$ we call it {\bf trace} of $A$ and denote it by $tr(A)$. 
	\end{enumerate}	
\end{definition}

\begin{remark}
{\em It should be clear that, according to the given definitions, noticing that $|(|T|)|=|T|$, we have 
\begin{equation}  
	||T||_1 = ||\:|T|\:||_1 = tr(|T|)\label{traceagg}
\end{equation} 
for every $T \in \gB_1(\sH)$ irrespective of the nature of $\bD$. This follows immediately from (c) of Proposition \ref{propTRACE1} and  selfadjointness of $|T|$.}
\end{remark}

With the said definition, an immediate corollary of Proposition \ref{propTRACE} follows. 
\begin{corollary}\label{corollary}
	Let $\sH$ be a real, complex or quaternionic Hilbert space. Then the following statements hold.
	\begin{enumerate}[(a)]
		\item $tr^{\bR} : \gB_1(\sH) \to \bR$ is  $\bR$-linear and $^*$-invariant.
		\item $tr^\bR(A)\ge 0$ for  $A\in\gB_1(\sH)$ positive.
		\item $tr^\bR(A)\ge tr^\bR(B)$ for $A,B\in\gB_1(\sH)$ satisfying $A\ge B$.
		\item $tr^{\bR}(AB)=tr^{\bR}(BA)$ for $A\in \gB_1(\sH)$ and $B\in \gB(\sH)$.
		\item  $tr^{\bR}(A)=tr(A)$ for $A^*=A\in \gB_1(\sH)$.
		\item If $A^*=A\in \gB_1(\sH)$ and $B^*=B\in \gB(\sH)$ then
		\begin{enumerate}[(i)]
			\item if $\bD=\bR, \bC$ then $tr^{\bR}(AB)=tr(AB)$ for $\bD=\bR,\bC$,
			\item if $\bD=\bH$, then $tr^\bR(AB)=tr_{N_A}(AB)$ where $N_A$ is the completion of a Hilbert basis of $Ker(A)^\perp$ made of eigenvectors of $A$.
		\end{enumerate}
		\item If $A^*=A\in \gB_1(\sH)$ and $P\in\cL(\sH)$, then $	tr^\bR(PA)=tr^\bR(PAP)=tr(PAP)$.
	\end{enumerate}
\end{corollary}
\begin{proof}
	The proofs of (a)-(f) follow easily from Proposition \ref{propTRACE}. To prove (g) just notice that $PAP$ is selfadjoint and that $tr^\bR(PA)=tr^\bR(PPA)=tr^\bR(PAP)$, which follows from $P=PP$.
\end{proof}
To conclude this list of properties of trace-class operators, we stress that, for complex Hilbert spaces, the following well known fact holds (e.g. see Proposition 4.41 in \cite{M}) which we extend here to the quaternionic case.  
\begin{proposition}\label{proptraceeq}
	If $\sH$ is a complex or quaternionic Hilbert space, the following asserts are equivalent for $A\in \gB(\sH)$.
	\begin{enumerate}[(i)]
		\item $A \in \gB_1(\sH)$,
		\item $\sum_{x\in N} |\langle x|A x\rangle| <+\infty$ for every Hilbert basis of $\sH$.
	\end{enumerate}
	If $\sH$ is real, then (i) implies (ii), but the converse implication is generally false in infinite dimensional spaces.
\end{proposition}

\begin{proof} 
	See Appendix \ref{Appendix}.
\end{proof}
 In \cite{V2} (ii) was used as definition of trace-class operator. We point out that this definition is equivalent to ours in the complex and quaternionic case only.

\subsection{A basis-independence property of $tr(A)$  for $\bD=\bH$.}
Contrarily to the real and complex case, in quaternionic Hilbert spaces $tr_N(A)$ may depend on the chosen Hilbert basis $N$ for a fixed $A \in \gB_1(\sH)$. However, an invariance property remains and, once more, it can be stated in terms of {\em real trace}. Although  this subject is not in the main stream of this work we spend some words about it. The reader interested in the application of previous results to Gleason's theorem can safely skip this section.

We remind  the reader that (Lemma 3.9 in \cite{GMP1}), if $J$ is an anti selfadjoint unitary operator in a quaternionic Hilbert space and we take $\imath \in \bH$ with $|\imath|=1$, then the complex vector space $\sH_{J\imath}:= \{z\in \sH \:|\: Jz=z\imath\}$ -- {\em using $\imath$ as imaginary unit} -- inherits the structure of complex Hilbert space when restricting the scalar product of $\sH$ to $\sH_{J\imath}$. Furthermore, if $N\subset \sH_{J\imath}$ is a Hilbert basis of $\sH_{J\imath}$, it is also a Hilbert basis of $\sH$ (Proposition 3.8  in \cite{GMP1}). A Hilbert basis $N$ of $\sH$ is evidently also a Hilbert basis of 
$\sH_{J\imath}$ whenever $N\subset  \sH_{J\imath}$, considering only  $\imath$-complex combinations of elements of $N$. If $N$ is a Hilbert basis of $\sH_{J\imath}$ and we
change $\imath$ to $\imath'$, 
then $N'= \{us \:|\: u\in N\}$ is a Hilbert basis of $\sH_{J\imath'}$ if $s^{-1}\imath s= \imath'$ and $|s|=1$.

Consider $A\in \gB(\sH)$ for a quaternionic Hilbert space $\sH$. The anti selfadjoint part $A-A^*$ is normal so, according to Theorem 5.9 in \cite{GMP1} (applied to $T=A-A^*$), its polar decomposition can be improved as follows. There exists $J\in \gB(\sH)$ with 
\begin{enumerate}[(1)]
	\item $J=-J^*$, $JJ= -I$,
	\item $A-A^*=J|A-A^*|$,
	\item $J$ commutes with both $A-A^*$ and $|A-A^*|$ (not necessarily with $A$).
\end{enumerate}

The polar decomposition theorem implies that $J$ is  {\em uniquely defined} on $Ker(|A-A^*|)^\perp = Ker(A-A^*)^\perp$\footnote{By definition of $|T|$, $||Tx||^2 = \langle x|T^*Tx\rangle =  \langle x| |T|^2x\rangle = |||T|x||^2$, so that $Ker(T)=Ker(|T|)$ if $T\in \gB(\sH)$.}.

We can now state and prove a proposition regarding a basis-invariance property in the quaternionic case.

\begin{proposition}\label{lastPROP} 
	Let $\sH$ be a quaternionic Hilbert space and $A\in \gB_1(\sH)$, let  $J$ satisfy (1)-(2) above with respect to $A-A^*$ and take $\imath \in \bH$ with $|\imath|=1$. For every Hilbert basis $N$ of $\sH$ such that $N \subset \sH_{J\imath}$ the following identity holds
	\begin{equation}\label{ivTRN}
		tr_N(A) = tr^\bR(A) + \frac{\imath}{2}tr(|A-A^*|)\:.
	\end{equation} 
	If $N'\subset \sH_{J'\imath}$ is another Hilbert basis of $\sH$, where  $Jx = J'x$ for $x\in Ker(|A-A^*|)^\perp$, then (\ref{ivTRN}) still holds for $N'$ in place of $N$.
\end{proposition}
\begin{proof} See Appendix \ref{Appendix}.
\end{proof}
We stress that other basis-independence properties of the trace  exist. Within the different approach of \cite{CGJ16}, for {\em every} $J=-J^{*}$ with $JJ=-I$ commuting with $A\in \gB_1(\sH)$,  $tr_N(A)$ is fixed  when the Hilbert basis $N$ varies in $\sH_{J\imath}$. (Such a $J$ does exist at least if $A$ is normal for Theorem 5.9 in \cite{GMP1}.) Our special choice of $J$ in Proposition \ref{lastPROP}  is always feasible for every $A\in \gB_1(\sH)$ though, in general, $J$ does not commute with $A$ and the invariant trace in $\sH_{J\imath}$ can be written in terms of real traces.

\section{A common statement of Gleason's theorem}
We are in a position to establish a common form of Gleason's theorem, valid for real, complex and quaternionic Hilbert spaces,  which uses the {\em real trace} instead of the {\em trace} appearing in Theorem \ref{teoG}. 

A {\bf convex body} is a subset $K\neq\emptyset$ of a real linear space $X$ such that $(1-\lambda)x+\lambda y\in K$ for every $x,y\in K$ and $\lambda\in [0,1]$. A point $\omega\in K$ said to be {\bf extremal} if it cannot be decomposed as $\omega = (1-\lambda) x +\lambda y$ with $\lambda\in (0,1)$ and $x, y\in K\setminus\{\omega\}$.

\subsection{Probability measures over $\cL(\sH)$}
\begin{definition}\label{defMU} 
	Let $\sH$ be a  real, complex or quaternionic Hilbert space. A
	$\sigma$-{\bf additive probability measure} over the lattice  $\cL(\sH)$ of orthogonal projectors over $\sH$  is a map  
	$\mu : \cL(\sH) \to [0,1]$ satisfying both 
	\begin{enumerate}[(i)]
		\item $\mu(I) = 1$,
		\item $\sum_{k\in K} \mu(P_k) = \mu\left( \vee_{k\in K} P_k\right)$,
	\end{enumerate}
	for $\{P_k\}_{k\in K} \subset \cL(\sH)$ with $K$ finite or countably infinite and $P_k \perp P_h$ for $k\neq h$. 

	We denote by $\cM(\sH)$ the convex body of $\sigma$-additive probability measures over  $\cL(\sH)$.
\end{definition}
\begin{remark}\label{remMU}  {\em $\null$
		\begin{enumerate}[(a)]
			\item Since $K$ is countable and $P_k P_h=0$ if $k\neq h$, the supremum  $\vee_{k\in K} P_k$ can always be computed as $\vee_{k\in K} P_kx = \sum_{k\in K}P_kx$ for every $x\in \sH$, the sum converging in the topology of $\sH$ if $K$ is infinite.  The proof is  identical in the three Hilbert space cases (see, e.g., \cite{M}, Theorem 7.22).
			\item $\cM(\sH)$ is to be understood as a convex subset of the set of functions $f:\cL(\sH)\rightarrow\bR$ which is clearly a real linear space with respect to the operations $(\alpha f+\beta g)(P):= \alpha f(P)+\beta g(P)$ for all $P\in\cL(\sH)$.
		\end{enumerate}
	 }
\end{remark}

\subsection{Fixing Varadarajan's statement of Gleason's theorem for $\bD=\bH$.}

We present and prove the statement of Gleason's theorem   completing the assertion with remarks about extremal measures. A further statement is added on the fact that, in spite of an apparent lack of information  embodied in the {\em real trace} if compared with the standard trace, probability measures are still able to distinguish between  elements of $\cL(\sH)$;
\begin{theorem}\label{GVt}  
	Let $\sH$ be a  real, complex or quaternionic separable Hilbert space with dimension greater than $2$. The following facts are true.
	\begin{enumerate}[(a)]
		\item  The class $\cM(\sH)$ of  $\sigma$-additive probability measures  $\mu$ over $\cL(\sH)$ is one-to one  with the class  $\cS(\sH)$  of  self-adjoint, positive,  operators $T\in \gB_1(\sH)$ with unit trace. This correspondence  is defined by the requirement  
		\begin{equation}\label{traceR}
			\mu(P) = tr^{\bR}(PT)\quad \mbox{ for all }P \in \cL(\sH)\:.
		\end{equation}
		That correspondence  preserves the real convex structures of $\cS(\sH)$ and $\cM(\sH)$\:.
		\item All orthogonal projectors  onto one-dimensional subspaces of $\sH$ belong to $\cS(\sH)$. They are precisely 
		the extremal element of $\cS(\sH)$ and are  one-to-one through (\ref{traceR}) with the extremal elements of $\cM(\sH)$.
		\item As a consequence of (\ref{traceR}), probability measures separate the elements of $\cL(\sH)$ (and evidently the elements of $\cL(\sH)$ separate the probability measures over $\cL(\sH)$)
	\end{enumerate}
\end{theorem}

\begin{proof} 
(a) (i) {\em Uniqueness of} $T$. For a given $\mu$, the associated $T \in \cS(\sH)$ is unique if exists. Indeed, if $tr^{\bR}(PT)= \mu(P) = tr^{\bR}(PT')$ for 
all $P\in \cL(\sH)$, restricting to  one-dimensional projectors $P=  \psi\langle \psi| \cdot \rangle$ we have 
$\langle \psi |(T-T') \psi\rangle =0$ for every vector $\psi \in \sH$. Since $T-T'$ is bounded and selfadjoint we get $T-T'=0$ (see Section \ref{sec1.1}, just below Point {\bf (C)}).

%if $\bD=\bR,\bC$. The same result arises from polarization identity, for $\bD=\bR$, using the fact that $T-T'$ is also self-adjoint.

(a)(ii) {\em Each $T\in \cS(\sH)$ defines a $\sigma$-additive probability measure $\mu \in \cM(\sH)$}. Let us now prove that selfadjoint positive operators $T\in \gB_1(\sH)$ with unit trace define probability measures (Definition \ref{defMU})  through (\ref{traceR}).  First focus on  $\sigma$-additivity. Taking Remark \ref{remMU} into account, this amounts to establish that 
$$
tr^{\bR}\left(\left(s\mbox{-}\sum_{k\in K} P_k\right)T\right) = \sum_{k\in K} tr^{\bR}(P_kT)
$$
for every class of orthogonal projectors $\{P_k\}_{k\in K}$ with $K$ finite or countably infinite and $P_kP_h=0$ for $h \neq k$ and where $s$- denotes the strong limit. Assuming $K=\bN$, define the orthogonal projectors $Q_n := \sum_{k=0}^nP_k$ and $Q :=  \mbox{s-}\lim_{n\to +\infty} Q_n = s\mbox{-}\sum_{k\in K} P_k$. What we have to prove is that $tr^{\bR}\left(QT \right) = \lim_{n\to +\infty} tr^{\bR}(Q_nT)$. 
Since the real trace is basis independent we can evaluate it on a (countable in our hypotheses)  Hilbert basis $N$ of $\sH$ obtained by completing  a Hilbert basis of
$Q(\sH)$ that, in turn, is  made of the union of Hilbert bases of each $P_n(\sH)$. Exploiting Corollary \ref{corollary} (g), we get
$
tr^{\bR}(Q_nT) =tr_N(Q_nTQ_n)
$
and similarly
$tr^{\bR}(QT) = tr_N(QTQ)$. 
What we have to prove 
is  that
\begin{equation}\label{equation}
\sum_{u\in N} \langle u | QTQ u\rangle = \lim_{n \to +\infty} \sum_{u\in N} \langle u| Q_nTQ_n u\rangle\:.
\end{equation}
Observe 
that, for every fixed $u\in N$,  we have $0\leq  \langle u | Q_{n}TQ_n u\rangle \leq  \langle u | QTQ u\rangle $ (indeed, if  $u\in P_m(\sH)$ then $\langle u | Q_{n}TQ_n u\rangle =0$ for $n<m$ and $\langle u | Q_{n}TQ_n u\rangle =  \langle u | QTQ u\rangle$ for $n\geq m$. If $u \in Q(\sH)^\perp$ both sides vanish). 
 Since $0\leq \langle u| Q_nTQ_n u\rangle \leq \langle u| QTQ u\rangle$ and $\sum_{u} \langle u | QTQ u\rangle < +\infty$, the dominated convergence theorem proves (\ref{equation}) completing the proof of $\sigma$-additivity.  Let us prove that the constructed measure ranges in $[0,1]$. So, take any $P\in \cL(\sH)$ and consider a Hilbert basis $N$ of $\sH$ completing a Hilbert basis $N_P$ of $P(\sH)$. Corollary \ref{corollary} (g) guarantees that $tr^{\bR}(PT) = tr_N(PTP)$. Now, exploiting the positivity of $T$ and Proposition \ref{propTRACE} (a)-(iii) we get
	\begin{equation*}
	0\le   tr_N(PTP)=\sum_{z\in N} \langle z|PTPz \rangle =\sum_{z\in N_P} \langle z|Tz \rangle \le \sum_{z\in N} \langle z|Tz \rangle =tr_N (T) =1\:,
	\end{equation*}
so that $tr^{\bR}(PT) \in [0,1]$.  Evidently $tr^{\bR}(IT)=1$. We proved that $T$ defines a $\sigma$-additive probability measure over $\cL(\sH)$ through (\ref{traceR}).

(iii) {\em Every $\mu\in \cM(\sH)$ can be written as in  (\ref{traceR}) for some $T\in \cS(\sH)$}. This is the hard part of the proof and we exploit the results established in \cite{V2} extending the original procedure by Gleason \cite{G}. Take a probability measure $\mu: \cL(\sH) \to [0,1]$. In view of its definition, for every Hilbert basis $N$, which is at most countable by hypothesis, we have 
$
1= \mu(I) = \sum_{u\in N}\mu(u\langle u|\cdot \rangle)
$
independently from the choice of $N$. Therefore, the map $f_\mu : \{ x\in \sH \:|\: ||x||=1\} \to [0,1]$ defined by $f_\mu(x) := \mu(x\langle x|\cdot \rangle)$ 
satisfies $\sum_{u\in N}f_\mu(u) =1$ independently from the choice of $N$. Such a map is called {\em frame-function} with weight $1$. As a consequence of Lemma 4.22 in \cite{V2}, since $f_\mu(x) \geq 0$ for every $x$, there is $T\in \gB(\sH)$ such that $T=T^*$ and 
$f_\mu(x) = \langle x|Tx\rangle$. In particular $T\geq 0$ because $f_\mu\geq 0$.
Since $1=  \sum_{u\in N}f_\mu(u)$ for every Hilbert basis $N$, we also have $1= \sum_{u\in N} \langle u|Tu\rangle$ which, from $T\geq 0$ can be re-written $1= \sum_{u\in N} \langle u||T|u\rangle$. According to Definition \ref{defdefT} we have proved that $T\in \gB_1(\sH)$ and also $tr(T)=1$ in view of (b) in Definition \ref{DEFTRV}.
It remains to prove that (\ref{traceR}) holds true. The identity $f_\mu(x) = \langle x|Tx\rangle$ can be rephrased to $\mu(\langle x| \cdot \rangle x) = \langle x|Tx\rangle$ for every $x \in \sH$ with $||x||=1$. If $P\in \cL(\sH)$, let $M$ be a Hilbert basis of $P(\sH)$ so that $Pz = \sum_{x\in M} x\langle x| z \rangle$ for every $z\in \sH$.
$\sigma$-additivity of $\mu$ implies $\mu(P) = \sum_{x\in M} \langle x|Tx \rangle$. Completing $M$ to a Hilbert basis $N$ of $\sH$, the found identity can be re-written
$\mu(P) = tr_M (PT)$. Taking the real part of both sides, we eventually get $\mu(P) = Re\left(tr_M (PT)\right) = tr^\bR(PT)$, where the last term no longer depends on $M$. The proof of (a) ends: the fact that  the correspondence between operators $T$ and measures $\mu$  preserves the real convex structures of $\cS(\sH)$ and $\cM(\sH)$ follows trivially from (a) in Corollary \ref{corollary}.

(b) We identify every $\mu\in \cM(\sH)$  with the corresponding $T\in \cS(\sH)$ according to (a).
  Consider  such a measure, that is $T\in  \gB_1(\sH)$ which is selfadjoint positive and $tr(T)=1$ and suppose that $T$ is not a one-dimensional orthogonal projector.  Now consider the spectral decomposition in Proposition \ref{propdecT},
$T = \sum_{u\in N}  u s(u)\langle u| \cdot \rangle $. There the  finite or countable orthonormal system of vectors $N$ (which is a Hilbert basis of $Ker(T)^\perp$) can be completed to a Hilbert basis of $\sH$ (by adding a Hilbert space of $Ker(T)$), the positive reals $s(u) \in (0,1]$ form the point spectrum of $T$ except for, possibly, the zero eigenvalue. Finally  they satisfy $\sum_{u\in N} s(u) =1$.  
If $T$ is not a one-dimensional orthogonal projector there are at least two different $u$, say $u_1$ and $u_2$ with  $s(u_1)>0$ and $1-s(u_1) \ge s(u_2) >0$. As a consequence, $T$  decomposes into the convex decomposition
$T= s(u_1) T_1 + (1- s(u_1)) T_2$ for 
$$
T_1 =  u_1 \langle u_1| \cdot \rangle \quad\mbox{and} \quad
T_2 :=  \sum_{u\neq u_1}   \frac{s(u)}{1- s(u_1)}u \langle u| \cdot \rangle\:.
$$ 
Notice that  (i) $T_1 \neq T_2$, (ii) $T_1, T_2 \neq 0$, (iii) $T_1,T_2 \in \gB_1(\sH)$ by construction, (iv) they are selfadjoint, (v)  $T_1,T_2 \geq 0$ and (vi) $tr(T_1)= tr(T_2)=1$, so $T_1$ and $T_2$ represent two different probability measures over $\cL(\sH)$. We conclude that $T$ cannot be extremal in the convex body of probability measures,
because it admits a non-trivial convex decomposition. 

To conclude the proof, we prove that  a one-dimensional orthogonal projector $P$ does not admit a non-trivial convex decomposition into a pair of different operators $T_1$ and $T_2$ representing probability measures over $\cL(\sH)$. So, suppose that $P = \lambda_1 T_1 + \lambda_2T_2$ with $T_1,T_2\in \gB_1(\sH)$ and $\lambda_1, \lambda_2 \in (0,1)$ such that $\lambda_1+\lambda_2 =1$.
Thus $P = PPP= \lambda_1 PT_1P + \lambda_2 PT_2P$.  As a consequence of  (d)(i) Proposition \ref{propTRACE1}, $PT_1P,  PT_2P\in \gB_1(\sH)$. In particular, if $P=  \psi \langle \psi | \cdot \rangle$, it must be  
$$
PT_rP =  \psi \langle \psi| T_r\psi \langle\psi|\cdot \rangle\rangle=   \psi \langle \psi|T_r\psi\rangle \langle \psi|\cdot\rangle = \psi q_r\langle \psi | \cdot \rangle\:,
$$ 
for $q_1,q_2 \ge 0$, the operators  $PT_rP$ being selfadjoint and positive.
Furthermore, completing $\psi$ to a Hilbert basis of $\sH$ and exploiting $tr(T_r) =1$ and the positivity of $T_r$ we see that $q_r=tr(PT_rP)\le 1$.
Finally, taking the trace of both sides of $P = \lambda_1 PT_1P + \lambda_2 PT_2P$ we have $\sum_i \lambda_iq_i =1$.
Summing up we have four reals $\lambda_1,\lambda_2,q_1,1_2$ such that
$$
\lambda_1,\lambda_2\in (0,1),\quad q_1,q_2\in [0,1],\quad \sum_{i=1}^2\lambda_i=1\quad\mbox{and}\quad \sum_{i=1}^2\lambda_i q_i=1\:.
$$
At this point we need the following result
\begin{lemma}\label{lemmatecnico}
	Let $2\le N\le \infty$ and $(p_n)_{n=0}^N\subset (0,1)$ and $(q_n)_{n=0}^N\subset [0,1]$ such that 
	$$
	\sum_{n=0}^N p_n=\sum_{n=0}^N p_nq_n=1\:,
	$$ 
	then it holds that $q_n=1$ for any choice of $n$.
\end{lemma}
\begin{proof}
See Appendix \ref{Appendix}.
\end{proof}
\noindent Exploiting Lemma \ref{lemmatecnico}, we get $q_1=q_2=1$. 
Since the operators $T_r$ are selfadjoint, positive, compact, trace-class and unit-trace, Proposition \ref{propdecT} yields
$
T_r=\sum_{z\in N_r}s_r(u) u \langle u|\cdot \rangle
$
where $N_r$ is a Hilbert basis of $Ker(T_r)^\perp$, $s_r(u)> 0$ and $\sum_{u\in N_3}s_r(u)=1$. Since also $PT_rP$ is unit-trace, completing $N_r$ to a Hilbert basis $N_r'$ of $\sH$ and exploiting Corollary \ref{corollary} (g) we get
\begin{equation*}
1=tr(PT_rP)=tr^\bR(PT_r)=\sum_{z\in N_r'} Re \langle z|PT_rz \rangle=\sum_{z\in N_r}s_r(u) \langle u|Pu \rangle =\sum_{u\in N_r}s_r(u)\|Pu\|^2\:.
\end{equation*}
Since $0\le \|Pu\|^2\le \|u\|^2=1$ and $\sum_{u\in N_r}s(u)=1$ we argue that $\|Pu\|=1$ for all $u\in N_r$ which is equivalent to $u\in P(\sH)$ for all $u\in N_r$. Since $P(\sH)$ is one-dimensional and generated by $\psi$ and the elements of $N_r$ are orthogonal to each other, it must be $N_r=\{u=\psi p_r\}$ for some $p_r\in\bH$ with $|p_r|=1$. At this point the proof would be complete, since $T_r=u\langle u|\cdot \rangle=\psi \langle \psi|\cdot\rangle=P$.
So, let us prove that $\|Pu\|=1$ for all $u\in N_r$. First, notice that if there exists $u_0\in N_r$ such that $s_r(u_0)=1$, then $\sum_{u\in N_r}s_r(u)=1$ and $s_r(u)> 0$ forces $N_r=\{u_0\}$ and so $1=s_r(u_0)\|Pu_0\|=\|Pu_0\|$. If, instead, all of the elements $u\in N_r$ satisfy $0<s_r(u)<1$, then the thesis follows from Lemma \ref{lemmatecnico} (notice that $N_r$ is at most countable).

(c) (Obviously, $\cL(\sH)$ separates the convex body of probability measures because they admit $\cL(\sH)$ as domain.) To prove that probability measures separate orthogonal projectors, suppose that $\mu(P)=\mu(P')$ for every probability measure $\mu$. As a consequence of (a),
$tr^{\bR}(PT)= tr^{\bR}(P'T)$ for every self-adjoint positive  $T\in \gB_1(\sH)$ with unit trace.  
Restricting to  one-dimensional projectors $T=  \psi\langle \psi| \cdot \rangle$, the written identity specialises to
$\langle \psi |(P-P') \psi\rangle =0$ for every vector $\psi \in \sH$. Since $P-P'$ is bounded, polarization identity proves $P-P'=0$
if $\bD=\bC,\bH$. The same result arises from polarization identity, for $\bD=\bR$, using the fact that $P-P'$ is also self-adjoint.
\end{proof}

\begin{remark}\label{REM} {\em $\null$ \begin{enumerate}[(a)]
		\item In case $\bD=\bR$ or $\bC$, the real trace in (\ref{traceR}) can be replaced by the standard trace without affecting the result because of (f) in Corollary \ref{corollary}, recovering the known statement of Gleason's theorem as in Theorem \ref{teoG}\:.
		\item As alternate possibilities in stating the theorem above, we observe that  right-hand side of (\ref{traceR}) satisfies  
	\beq
		tr^{\bR}(PT)=tr^{\bR}(TP)= tr^{\bR}(PTP) = tr(PTP) \label{tracerest}
		\eeq 
		which follows from Corollary \ref{corollary} (f) (notice that $tr^{\bR}(TP)= tr^{\bR}(TPP)= tr^{\bR}(PTP)$ since $P=PP$ and the real trace is cyclic).
%		The first identity arises from (b) in Corollary \ref{corollary}, the second and the third are consequences of the fact that $tr(PTP)$ is well-defined since $PTP$ is self-adjoint and thus it can be computed in particular on a Hilbert basis $N$ completing a Hilbert basis of $P(\sH)$, obtaining $tr(PTP) = tr_N(TP)$. Taking the real part of both sides we have
%		$tr(PTP)= tr^{\bR}(PTP) = tr^{\bR}(TP)$.
\item  The sequence of  identities (\ref{tracerest})  proves  that the  standard trace can be used  in place of the real trace in some elementary physical computations involving  orthogonal projectors and statistical operators. A price to pay is  that the formulas have to be re-arranged in order to always deal  with {\em selfadjoint} arguments of the trace as  $\mu(P) = tr(PTP)$ (the other possibility $\mu(P) =
tr\left[\frac{1}{2}(PT+PT)\right]$ would be  a trivial rephrasing of $\mu(P) =tr^\bR(PT)$). However, passing to physical objects more complicated than probabilities as  the expectation value or the standard deviation of  (generally unbounded) observables (see identites (i) and (ii) in (a) and (b) of Proposition \ref{propAT} below),  this route would become technically very complicated and unnatural. As an example, replacing $P$ for  a (generally  unbounded)  selfadjoint operator $A$, the sequence of identities $tr^{\bR}(AT)=tr^{\bR}(TA)= tr^{\bR}(ATA) = tr(ATA)$ would result trivially false. Regardless subtle problems related with domains, the crucial obstruction is that $AA \neq A$.  Furthermore, even starting from $\mu(P^{(A)}_E)= tr(P^{(A)}_ETP^{(A)}_E)$ where $P^{(A)}_E$ is the projection-valued   measure of the selfadjoint operator $A$, the {\em quadratic} $P^{(A)}_E$-dependence  of $tr(P^{(A)}_ETP^{(A)}_E)$ makes impossible to take advantage of the spectral integration procedure to achieve  results similar to those asserted  in Proposition \ref{propAT} below regarding the expectation value and the standard deviation of $A$.
\end{enumerate}}
\end{remark}

\section{Compatibility with the standard physical formalism}
In physical applications in complex Hilbert spaces the standard notion of {\em trace}, instead of {\em real trace} 
is exploited. This section  proves that the notion of real trace is completely enough to deal with physical formalism and its use provides a common mathematical tool valid for real, complex and quaternionic formulations.

With the hypotheses of Theorem \ref{GVt}, extremal probability measures over $\cL(\sH)$ are called {\bf pure states} in the language of physicists. According to the statement of Gleason's theorem, they are represented by all operators of the form $T = \psi \langle \psi|\:\:\rangle$ for any unit vector $\psi \in \sH$ fixed up to a scalar $q\in \bD$ with $|q|=1$. This is because, as the reader can immediately prove, that is the general form of orthogonal projectors onto one-dimensional subspaces spanned by unit vectors $\psi$.
The remaining operators $T \in \cS(\sH)$ describing generic probability measures over $\cL(\sH)$ according to the theorem above, are called {\bf mixed states} or also {\bf statistical operators}. Pure and mixed states are called {\bf quantum states}.

\subsection{Observables}
According to the standard terminology of physicists, an {\bf observable} is a selfadjoint (generally unbounded) operator $A: D(A) \to \sH$, where $D(A)$ is a dense subspace in the Hilbert space describing a given quantum system. We henceforth assume it to be real, complex or quaternionic. {\bf Elementary observables} (also called {\bf elementary propositions}) are in particular $P \in \cL(\sH)$.
 $A$ can be spectrally decomposed 
$$A = \int_{\sigma(A)} \lambda dP^{(A)}(\lambda)\:,$$
where we have introduced the {\em projector-valued measure} (PVM) 
$$
\left\{P^{(A)}_E\right\}_{E\in \cB(\sigma(A))}\subset \cL(\sH)
$$
of $A$ over the Borel $\sigma$-algebra $\cB(\sigma(A))$ defined over the spectrum $\sigma(A)$ of $A$.
 That is a (closed) subset of $\bR$ because $A=A^*$. See, e.g., \cite{M} for the complex case, \cite{MO1} for the real case, and $\cite{GMP2}$ for the quaternionic case, noticing that $\sigma(A)$ should  be interpreted as the intersection of the {\em spherical spectrum} and a complex slice $\sigma(A)= \sigma_S(A) \cap \bC^+_i$. However as $A=A^*$, the spherical spectrum is completely included in $\bR$ and  this intersection turns out to be a subset of $\bR$ independent of the chosen complex slice $\bC_i \subset \bH$.
 The PVM is uniquely associated with $A$ and satisfies
 \begin{equation}\label{PVM}
 P^{(A)}_{\sigma(A)}=I,\quad P^{(A)}_{E\cap F}=P^{(A)}_E P^{(A)}_F\quad \mbox{and} \quad P^{(A)}_{\cup_{n=1}^\infty E_n}x=\sum_{n=1}^\infty P^{(A)}_{E_n}x
 \end{equation}
 for any $x\in\sH$ and for any $E,F,E_n\in \cB(\sigma(A))$ with the subsets $E_n$ pairwise disjoint.

If $A:D(A)\to \sH$ is selfadjoint and  $f : \bR \to \bR$ is a Borel-measurable function, the selfadjoint operator  
\begin{equation} 
f(A): = \int_{\sigma(A)} f(\lambda) dP^{(A)}(\lambda) \label{fA}
\end{equation}
is defined with domain
\begin{equation}\label{DomfA}
D(f(A)) = \left\{x\in \sH \:\left|\: \int_{\sigma(A)} |f(\lambda)|^2 d\mu^{(A)}_x(\lambda) <+\infty\right.\right\}\:,
\end{equation}
where we introduced the finite positive Borel measure associated with $x$ and $A$:
$$
\mu^{(A)}_x : \cB(\sigma(A)) \ni E \mapsto || P^{(A)}_E x||^2 \in  [0, ||x||^2]\:.
$$ 
Moreover, for any $x\in D(f(A))$ it holds that
\begin{equation}\label{equationf}
\langle x|f(A)x \rangle=\int_{\sigma(A)}f(\lambda) d\mu^{(A)}_x\quad\mbox{and}\quad \|f(A)x\|^2=\int_{\sigma(A)}|f(\lambda)|^2 d\mu^{(A)}_x\:.
\end{equation}

For $\bD=\bC$ and $\bD=\bH$ a larger class of measureable functions $f$ can be used (\cite{M,GMP2}), but we stick here to the real-valued ones because  they are completely sufficient for the rest of the work.

Physically speaking, the Borel sets $E\in \cB(\sigma(A))$ are the outcomes of measurement procedures of $A$, and $P^{(A)}_E$ is the elementary propositions  corresponding to the statement ``the measurement of $A$ belongs to $E$''.  More precisely, if $T$ is a {\em quantum state} over the Hilbert space $\sH$ satisfying the hypotheses of Theorem \ref{GVt}, the natural interpretation of 
\begin{equation}\label{muaT}
\mu_T^{(A)} : \cB(\sigma(A))\ni E \mapsto 
tr^{\bR}(P_E^{(A)}T)\in [0,1]
\end{equation}
is the probability to obtain $E$ after a measurement of $A$ in the quantum state $T\in \cS(\sH)$. In particular, if $T$ is {\em pure}, so that $T= \psi\langle \psi| \cdot \rangle$ for some unit vector $x \in \sH$, we have
\begin{equation}
\mu_T^{(A)}(E) = ||P_E^{(A)}\psi||^2  = \mu^{(A)}_\psi(E)\:.
\end{equation}
The proof is trivial, just complete $\{\psi\}$ as a Hilbert basis of $\sH$ and compute the real trace along that basis. This way, the {\bf expectation value} of $A$ with respect to the state $T$ can be defined
$$\langle A \rangle_T := \int_{\sigma(A)} \lambda \: d\mu_T^{(A)}(\lambda) \:,$$
provided the function $\sigma(A) \ni \lambda \to \lambda \in \bR$ is $\cL^1(\sigma(A),\mu_T^{(A)})$. Similarly, the {\bf standard deviation} is defined as 
  $$\Delta A_T  :=  \sqrt{\int_{\sigma(A)} (\lambda - \langle A \rangle_T)^2\:  d\mu_T^{(A)}(\lambda)} =   \sqrt{\int_{\sigma(A)} \lambda^2 \:  d\mu_T^{(A)}(\lambda)-  \langle A \rangle^2_T}\:,$$
provided $\sigma(A) \ni \lambda \to \lambda \in \bR$ is  $\cL^2(\sigma(A),\mu_T^{(A)})$.
(Notice $\cL^2(\sigma(A),\mu_T^{(A)}) \subset \cL^1(\sigma(A),\mu_T^{(A)})$ since the measure is finite.)

The next proposition establishes that, with our statement of Gleason's theorem, the usual formal results handled by physicists (see formulas in (b)-(d) below) are however valid when 
$\bD=\bR, \bC, \bH$. It happens  if systematically replacing the standard trace with $tr^{\bR}$ and assuming natural conditions on the states\footnote{Weaker necessary and sufficient conditions assuring  that   these formulas are valid can be found in \cite{M} in the complex case, referring to the Hilbert-Schmidt class of the operators we do not consider here.}.  Referring to domain issues in (b) and (c) below we observe that  from (\ref{DomfA}) we have  $D(A^2)\subset D(A)=D(|A|)$.

\begin{proposition} \label{propAT}
	Let $\sH$ be a real, complex or quaternionic Hilbert space satisfying the hypotheses of Theorem \ref{GVt}, $T\in \cS(\sH)$ a quantum state and $A: D(A) \to \sH$, densely defined, an observable (i.e., $A=A^*$). The following facts hold.
	\begin{enumerate}[(a)]
		\item $\mu_T^{(A)}$ as in (\ref{muaT}) is a well-defined probability measure over $\cB(\sigma(A))$.
		\item If  $Ran(T) \subset D(A)$   and  $|A|T \in \gB_1(\sH)$ (always valid if $A \in \gB(\sH)$), then 
		\begin{enumerate}[(i)]
			\item $\langle A \rangle_T$ is defined,
			\item $\langle A \rangle_T = tr^{\bR}(AT)$.
		\end{enumerate}
		\item  If  $Ran(T) \subset D(A^2)$   and  $|A|T, A^2T \in \gB_1(\sH)$ (always valid if $A \in \gB(\sH)$), then 
		\begin{enumerate}
			\item $\Delta A_T$ is defined,
			\item $\Delta A_T = \sqrt{tr^{\bR}(A^2T)- \left(tr^{\bR}(AT)\right)^2}$.
		\end{enumerate}
		\item Assume that $T= \psi\langle \psi| \:\rangle$  with $||\psi||=1$
		\begin{enumerate}[(i)]
			\item If $\psi \in D(A)$ then (b) is valid and $\langle A \rangle_T = \langle \psi|A \psi \rangle$,
			\item If $\psi \in D(A^2)$ then (c) is valid and $\Delta A_T = \sqrt{\langle \psi|A^2 \psi \rangle - \langle \psi|A \psi \rangle^2}$.
		\end{enumerate}
	\end{enumerate}
\end{proposition}

\begin{proof} (a) Taking the definition of projector valued measure into account, the proof is a trivial re-adaptation of the part (a)(ii) of the proof of Theorem \ref{GVt}.
	
%In the rest of the proof, we will exploit the spectral theorem for selfadjoint operators indifferently for $\bD=\bC$ (see, e.g. \cite{M}), $\bD=\bR$ (see, e.g., \cite{MO1}) and $\bD=\bH$ (using in this case the formulation presented in \cite{GMP2} noticing that the spherical spectrum of $A$ is a subset of $\bR$ and thus $\bC^+_i \cap \sigma_S(A)=: \sigma(A)$ does not depend on $\bC^+_i$).\\
(b)(i) We prove the thesis using the weak hypotheses (1) and (2), since they are automatically true if $A\in \gB(\sH)$. As already stressed,  $D(|A|) = D(A)$ so $Ran(T) \subset D(A)=D(|A|)$ is true and both $AT$, $|A|T$ are well defined with the said hypotheses. Next the polar decomposition theorem for (generally unbounded) self-adjoint operators gives  $A= U|A|$ with $|A|$ and $U := \mbox{sign}(A)\in\gB(\sH)$ defined as in (\ref{fA})). As a consequence,  $AT = U|A|T\in \gB_1(\sH)$ because $U\in \gB(\sH)$ and $\gB_1(\sH)$ is two-sided ideal.
 Now, referring to the Borel $\sigma$-algebra over $\sigma(A) \subset \bR$ we can construct \cite{M} a sequence of real simple functions 
 $$
 s_n= \sum_{i_n\in\cI_n}c^{(n)}_{i_n}\chi_{E_{i_n}^{(n)}}:  \sigma(A) \to \bR\quad \mbox{with } c_{i_n}^{(n)}\in\bR,\ \mbox{and}\ \cI_n\mbox{ finite}
 $$ 
 which satisfies
 \begin{equation}\label{monot}
 0\leq |s_n| \leq |s_{n+1}| \leq |id|\:, \quad s_n \to id \quad \mbox{point-wise for $n\to +\infty$,}
 \end{equation}
 where $id : \sigma(A) \ni \lambda  \mapsto \lambda \in \bR$.
By direct application of the given definitions, if 
$$
A_n := \int_{\sigma(A)} s_n dP^{(A)}=\sum_{i_n\in\cI_n}c^{(n)}_{i_n}P^{(A)}_{E_{i_n}^{(n)}} \in \gB(\sH)\:,
$$ 
exploiting (\ref{equationf}), the monotone and  the dominated convergence theorems, we have both
\begin{equation}\label{convA} 
\langle \psi| A_n \psi\rangle \to \langle \psi |A \psi\rangle\:, \quad 
\langle \psi| |A_n| \psi\rangle \to \langle \psi ||A| \psi\rangle
\quad \forall \psi \in D(A)\quad\mbox{as $n \to +\infty$}
\end{equation}
and also
\begin{equation}
|\langle \psi| A_n \psi\rangle | \leq \langle \psi| |A_n| \psi\rangle \leq  \langle \psi| |A| \psi\rangle \:.\label{convA2}
\end{equation}
On the other hand, if $M$ is a Hilbert basis of $\sH$ obtained by completing a Hilbert basis $N$ of $Ker(T)^\perp$ made of eigenvectors of $T$ according to Proposition \ref{propdecT}, exploiting Corollary \ref{corollary} (f)(ii) and Proposition \ref{propTRACE} (e) we have both
\begin{equation}\label{convAA'}
\begin{split}
tr_M(A_nT) &= tr_M\left(\sum_{i_n\in\cI_n}c^{(n)}_{i_n} P^{(A)}_{E^{(n)}_{i_n}}T\right)=\sum_{i_n\in\cI_n}c^{(n)}_{i_n} tr_M(P^{(A)}_{E^{(n)}_{i_n}}T)=\sum_{i_n\in\cI_n}c^{(n)}_{i_n} tr_M(TP^{(A)}_{E^{(n)}_{i_n}})=\\
&=\sum_{i_n\in\cI_n}c^{(n)}_{i_n} tr^\bR(TP^{(A)}_{E^{(n)}_{i_n}})=\sum_{i_n\in\cI_n}c^{(n)}_{i_n} \mu_T(E^{(n)}_{i_n})=\int_{\sigma(A)}s_n d\mu_T^{(A)}
\end{split}
\end{equation}
and similarly
\begin{equation}\label{convAA}
 tr_M(|A_n| T) = \int_{\sigma(A)} |s_n| \: d\mu_T^{(A)}\:.
\end{equation}
Looking at the identity (\ref{convAA}), by monotone convergence theorem, for $n\to +\infty$,
$$
tr_M(|A_n| T)=\int_{\sigma(A)} |s_n|(\lambda) \: d\mu_T^{(A)}(\lambda) \to \int_{\sigma(A)} |\lambda| \: d\mu_T^{(A)}(\lambda)\:,
$$
and simultenously we have
$$tr_M(|A_n|T) = \sum_{u\in N} s(u) \langle u| |A_n| u\rangle \to \sum_{u\in N} s(u) \langle u| A u\rangle = tr_M(|A|T)\:,$$
 where $s(u)\geq 0$ are the eigenvalues of $T$,
again by monotone convergence theorem and (\ref{convA}). Putting all together and taking the real part, we get
$$
tr^{\bR}(|A|T) =  \int_{\sigma(A)} |\lambda| \: d\mu_T^{(A)}(\lambda)\:.
$$
We have in particular established that the integral in the right-hand side is finite (because the left-hand side exists by hypothesis) and thus 
$\langle A\rangle_T$ is well defined.

(ii) Let us look at the identity in (\ref{convAA'}). From the dominated convergence theorem
taking (\ref{monot}) into account, we obtain for $n\to \infty$
$$ 
tr_M(A_nT)=\int_{\sigma(A)} s_n(\lambda) \: d\mu_T^{(A)} \to  \int_{\sigma(A)} \lambda \: d\mu_T^{(A)}\:.
$$
On the other hand, 
$$tr_M(A_nT) = \sum_{u\in N}  \langle u| A_n u\rangle s(u) \to \sum_{u\in N}  \langle u| A u\rangle s(u)= tr_M(AT)\:,$$
where we have once again applied the dominated convergence theorem as is permitted by (\ref{convA2}).
Putting all together and taking the real part we get
$$tr^{\bR}(AT) =  \int_{\sigma(A)} \lambda \: d\mu_T^{(A)}(\lambda) =: \langle A\rangle_T\:,$$
concluding the proof of (ii).

(c) The proof is strictly analogous to that of (b) also noticing that the hypotheses of (c) implies those of (b) and that $\cL^2(\sigma(A),\mu_T^{(A)}) \subset \cL^1(\sigma(A),\mu_T^{(A)})$ because $\mu_T^{(A)}$ is finite.\\
(d) The thesis consists of trivial subcases of (b) and (c) in particular completing $\{\psi\}$ to a Hilbert basis of $\sH$ to be used to explicitly compute the various traces.
\end{proof}

\subsection{Symmetries}
Symmetries (including time evolution) of a quantum system described on a Hilbert space $\sH$ can be represented in terms of  various transformations of mathematical structures entering the game (see \cite{Landsman,M} for  exhaustive surveys when $\bD=\bC$). We only say that, when the set of elementary observables consists of the whole $\cL(\sH)$ (absence of superselection rules and gauge group), symmetries are represented by unitary operators $U : \sH \to \sH$ {\em defined up to signs for $\bD= \bR, \bH$};  they are  represented by either unitary or anti-unitary operators $U : \sH \to \sH$ {\em defined up to phases for $\bD=  \bC$} (see \cite{V2}). Symmetries act on quantum states in a standard way: $T \to UTU^{-1}$, where it is easy to prove that $UTU^{-1}$ is still a quantum state if $T$ is. It is more strongly evident that, for a fixed $U$,  the map $T \to UTU^{-1}$ defines an {\em automorphism} of the space of the  quantum states $\cS(\sH)$  viewed as a convex body in the real vector space $\gB_1(\sH)$.

 For $\bD =\bR$ or $\bC$ that action on states has a corresponding {\em dual action} on observables according to the real and complex formulation of Gleason's theorem as presented in Theorem \ref{teoG}. There, the trace is cyclic so that
\begin{equation}\label{dual}
tr(P\:UTU^{-1})= tr(U^{-1}PU\:T)
\end{equation}
and, evidently, $U^{-1}PU \in \cL(\sH)$ if $P\in \cL(\sH)$. 
It is easy to see that, for a fixed $U$,  this map defines an {\em automorphism} of the orthocomplemented lattice $\cL(\sH)$
(all this surely holds for the unitary case and it can be proved to hold also in the antiunitary case for $\bD=\bC$).

 Summing up, symmetries induced by (anti) unitaries $U : \sH \to \sH$ can be viewed as automorphisms of the real convex body of the states or, alternatively, of the orthomodular lattice of (elementary) observables: 
$$
\cS(\sH) \ni  T \to UTU^{-1} \in \cS(\sH) \:,  \quad \cL(\sH) \ni P \to U^{-1}PU \in \cL(\sH)\:.
$$
and this duality interplay corresponds to the physical fact that, looking at measurements,  the result of the action of a symmetry on a state can be obtained  by a corresponding dual action on observables keeping fixed the state. Everything is encoded in (\ref{dual}), and that identity is responsible for several crucial theoretical tools in theoretical physics, like the duality between {\em Schr\"odinger picture} and  {\em Heisenberg picture} dealing with time evolution.  All that holds for $\bD=\bR, \bC$.  Due to (f)(i) in Corollary \ref{corollary}, with this choice of $\bD$, (\ref{dual}) can be equivalently  re-written without loss of information replacing $tr$ for  $tr^{\bR}$, since $P$, $UTU^{-1}$ and  $T$, $U^{-1}PU$ are pairs of self-adjoint operators. 

The next straightforward result states that (\ref{dual}) survives the extension to the quaternionic formulation, provided $tr$ is replaced for $tr^{\bR}$ according to Theorem \ref{GVt}. More generally:
\begin{proposition} 
	Let $\sH$ be a real, complex or quaternionic Hilbert space  and $A \in \gB_1(\sH)$, $B \in \gB(\sH)$ or {\em vice versa}. It holds that
	\begin{equation}\label{dual2}
		tr^{\bR}(A\:UBU^{-1})= tr^{\bR}(U^{-1}AU\:B)
	\end{equation}
for every operator $U: \sH \to \sH$ which is unitary if $\bD=\bR, \bH$ or indifferently unitary or anti-unitary if $\bD=\bC$.

In particular, the thesis is true if  $A=P\in \cL(\sH)$ and $B=T \in \cS(\sH)$.
\end{proposition}
\begin{proof} This is an immediate consequence of (d)(ii) Proposition \ref{propTRACE1} and (d)(ii)  Proposition \ref{propTRACE}.
\end{proof}
As the last result, we prove that the action of a continuous symmetry makes the probabilities computed through $tr^{\bR}$ continuous as well, as it is expected from physics\footnote{(Continuous) groups of symmetries are more generally described in terms of  unitary projective representations  due to the "phase" ambiguity in associating unitaries to symmetries (see, e.g., Ch.12 of  \cite{M} for the complex case and \cite{V2} for the general case),  however we only stick here to the more elementary case.}.
 \begin{proposition}
 Let $\sH$ be a real, complex or quaternionic Hilbert space, $A \in \gB(\sH)$ and $B^*=B\in \gB_1(\sH)$ with $B\geq 0$ and $G \ni g \mapsto U_g$ a strongly continuous unitary representation of the topological group $G$. Then
the map $$G \ni g \mapsto tr^{\bR}(A\:U_gBU_g^{-1})$$ is continuous.
In particular the thesis is true if  $A=P\in \cL(\sH)$ and $B=T \in \cS(\sH)$.
\end{proposition}

\begin{proof} We know that  $tr^{\bR}(A\:U_gBU_g^{-1}) = tr^{\bR}(U_g^{-1}AU_g\:B)$ so we prove the thesis in this second form.
It is clear that, as $U_gAU_g^{-1}\in \gB(\sH)$  if $A \in \gB(\sH)$ and $U_{gg'}= U_gU_{g'}$, continuity at $g$ is equivalent to continuity at the neutral element $e$ of $G$. Let us prove it.  If $N\subset \sH$ is a Hilbert basis,
$ tr^{\bR}(U_g^{-1}AU_g\:B)=  Re\left( tr_N(U_g^{-1}AU_g\:B)\right)$.
To conclude, it is enough proving that the right-hand side tends to  $Re\left( tr_N(U_e^{-1}AU_e\:B)\right) = Re\left( tr_N(A\:B)\right)$
 as $g \to e$ with a suitable choice of the  basis $N$.  It is convenient to chose $N$ as the basis  of eigenvectors $u$ of $B$ that exists for Proposition \ref{propdecT} (completing the basis of $Ker(B)^\perp$ by adding a Hilbert basis of $Ker(B)$). Non-vanishing eigenvalues are real numbers $s(u)>0$ (because $B\geq 0$) and
$tr_N(U_g^{-1}AU_g\:B)= \sum_{u} \langle u| U_g^{-1}AU_g u\rangle s_u = \sum_{u} s_u || AU_g u||^2$. Since $B=|B|$, we have 
$\sum_u s(u) = ||B||_1 <+\infty$ from (\ref{traceagg})  and $ || AU_g u||^2 \leq ||A||\:||U_g||\: ||u||\leq ||A||$, the dominated convergence theorem proves that
$tr_N(U_g^{-1}AU_g\:B) \to tr_N(U_e^{-1}AU_e\:B)$ as $g\to e$, concluding the proof using the fact that $Re : \bD \to \bR$ is continuous.
\end{proof}

\section*{Acknowledgments} The authors thank  J. Gantner for drawing their attention on the subject and A. Perotti for pointing out \cite{T}. 

\appendix
\section{Proof of some propositions}\label{Appendix}

\begin{proof}[{\bf Proof of Proposition \ref{propTRACE1}}]
	In the complex case all statements are established, e.g., in Theorems 4.31 and 4.34, including Remark 4.35, of \cite{M}. The proofs for the real case are identical. What we need to prove is that (a)-(e) are valid for the quaternionic case, too. The route we follow is a reduction procedure. Observe that, if $(\sH, \langle\:|\:\rangle)$ is a quaternionic Hilbert space, the set $\sH$ becomes a real Hilbert space when considering only real linear combinations among the quaternionic ones and making use of the real symmetric scalar product $(x|y):= Re \langle x| y\rangle$ for $x,y \in \sH$. In particular, completeness survives this change of viewpoint because $||x|| = \sqrt{\langle x| x\rangle} = \sqrt{( x| x)}$ if $x\in \sH$. In the rest of the proof, the real afore-mentioned Hilbert space constructed out of the quaternionic Hilbert space $\sH$ will be denoted by $\sH_\bR$.
We stress that as sets $\sH = \sH_\bR$, the difference stays in the structures over these sets. The following lemma will be useful.
\begin{lemma}\label{lemma}
	Let $(\sH, \langle \:|\:\rangle)$ be a quaternionic Hilbert space and $(\sH_\bR, (\:|\:))$ the associated real Hilbert space as said above. The following facts are true.
	\begin{enumerate}[\bf (\ref{lemma}a)]
		\item $\gB(\sH) \subset \gB(\sH_\bR)$ is exactly made of the $\bR$-linear operators in $\gB(\sH_\bR)$ commuting with 
		the three $\bR$-linear maps $\sH_\bR \ni x \mapsto xi\in \sH_\bR$, $\sH_\bR \ni x \mapsto xj\in \sH_\bR$ and $\sH_\bR \ni x \mapsto xk \in\sH_\bR$.
		\item if $N \subset \sH$ is a quanternionic Hilbert basis, then the associated set $$N_{\bR}:= \{u, ui,uj, uk \:|\:u \in N\}$$ is  a real Hilbert basis of $\sH_\bR$.
		\item If $A \in \gB(\sH)$, the adjoint $A^* \in \gB(\sH)$ is also the adjoint in $\gB(\sH_\bR)$.
		\item If  $\gB(\sH) \ni A \geq 0$ then $A^*=A \geq 0$ as operator in  $\gB(\sH_\bR)$ and the squared root $\sqrt{A} \in \gB(\sH)$  coincides with that computed in $\gB(\sH_\bR)$.
		\item If $A\in \gB(\sH)$, $|A| \in \gB(\sH)$ coincides with the absolute value computed in $\gB(\sH_\bR)$.
		\item  If $A\in \gB(\sH)$, $A \in \gB_\infty(\sH)$  if and only if $A \in \gB_\infty(\sH_\bR)$.
	\end{enumerate}
\end{lemma}
\begin{proof}
(\ref{lemma}a) is evident. (\ref{lemma}b) is true because $N_{\bR}$ is $(\:|\:)$-orthonormal  and 
$$
||x||^2 = \sum_{u \in N} |\langle u|x\rangle|^2 = 
 \sum_{u \in N} |( u|x)|^2 + |(ui|x)|^2 +|(uj|x)|^2+|(uk|x)|^2\quad\mbox{if}\quad x \in \sH = \sH_\bR\:.
$$

(\ref{lemma}c) If $A \in \gB(\sH)$, the adjoint  $A^* \in \gB(\sH)$ is also the adjoint in $\gB(\sH_\bR)$ because 
$\langle A^*x|y\rangle = \langle x|Ay\rangle$ for $x,y \in \sH$ implies $(A^*x|y) =( x|Ay)$ for $x,y \in \sH_\bR= \sH$ and this identity completely defines the adjoint operators for elements of $\gB(\sH_\bR)$. 

(\ref{lemma}d) and (\ref{lemma}e) If  $\gB(\sH) \ni A \geq 0$ then
$A^*=A$. Moreover, since $A\ge 0$, it holds in particular that $\langle x|Ax\rangle \in\bR$ and so $(x|Ax)=\Re\langle x|A x\rangle  =\langle x|Ax\rangle\ge 0$ for all $x\in\sH=\sH_\bR$, so $A^*=A\geq 0$  in $\gB(\sH_\bR)$ and, there, $A$ admits positive squared root. By uniqueness of the positive squared root, $\sqrt{A} \in \gB(\sH)$ of (self-adjoint) positive operators   coincides with that computed in $\gB(\sH_\bR)$. Consequently, $|A| = \sqrt{A^*A}\in \gB(\sH)$ coincides with the absolute value computed in $\gB(\sH_\bR)$.

(\ref{lemma}f)  $A\in \gB_\infty(\sH)$ means by definition that every sequence $\{Ax_n\}_{n\in \bN}$ admits a converging subsequence if 
 $\{x_n\}_{n \in \bN} \subset \sH$ is bounded. Since the norms of $\sH$ and $\sH_\bR$ coincides, then $\sH=\sH_\bR$ as a set, and $A \in \gB(\sH_\bR)$. The thesis follows. 
\end{proof}

(a) (For $\sH$ quaternionic.) Suppose that (\ref{defT}) is true. Passing to the Hilbert basis $N_\bR$ of $\sH_\bR$ and using  (\ref{lemma}b) and  (\ref{lemma}e) we have
\begin{equation}\label{4traces} 
\sum_{u \in N} \langle u||T|u\rangle = \frac{1}{4} \sum_{v \in N_\bR} ( v||T|v)\:.
\end{equation}
Since (a) is valid in real Hilbert spaces, (\ref{defT})  must be in particular valid for any other basis $N'_\bR$ constructed out of a Hilbert basis $N'$ of $\sH$ and
$$+\infty > \sum_{u \in N} \langle u||T|u\rangle = \frac{1}{4} \sum_{v \in N_\bR} ( v||T|v) 
= \frac{1}{4} \sum_{v' \in N'_\bR} ( v'||T|v') = \sum_{u' \in N} \langle u'||T|u'\rangle\:.$$
We finally observe that for $A \in \gB(\sH)$, identity (\ref{4traces}) and the validity of (a) in the real case prove  $A \in \gB_1(\sH)$ if and only if $A \in \gB_1(\sH_\bR)$.

(b) (For $\sH$ quaternionic.)
Suppose that $A \in \gB_1(\sH)$, we prove that (i) and (ii) are true. The proof of (a) implies  $A \in \gB_1(\sH_\bR)$, so that, since (b) is valid in the real case, $A \in \gB_\infty(\sH_\bR)$. Finally (\ref{lemma}f) implies that $A \in \gB_\infty(\sH)$ and (i) holds. The validity of (ii) immediately follows from Proposition \ref{propdecT} since $|A|$ is compact because $A$ is (Proposition \ref{prop1}). To conclude, let us prove that (i),(ii) imply $A \in \gB_1(\sH)$. This immediately arises by applying Proposition \ref{propdecT} ($|A|$ is compact because $A$ is) and proving that (\ref{defT}) is true for a Hilbert basis of $\sH$ which completes the Hilbert basis of $Ker(A)^\perp$ existing for the said proposition. The proof of (b) is over. 

(c) (For $\sH$ quaternionic.) Again, the identity arises by direct application of Proposition \ref{propdecT} and (a), (b).

(d) (For $\sH$ quaternionic.) (i) Suppose that $a,b \in \bR$ and $A,B \in \gB_1(\sH)$. $A,B \in \gB_1(\sH_\bR)$ form the last assertion in the proof of (a) above. Since (i) holds true in the real case, $aA+bB \in \gB_1(\sH_\bR)$. The last statement in the proof of (a) proves that  
$aA+bB \in \gB_1(\sH)$. (ii) can be proved similarly exploiting (\ref{lemma}c) and the validity of (ii) in the real case. (iii) We already know that $\gB_1(\sH_\bR)$ is Banach referring to the corresponding norm here denoted by $||\:||_1^{\bR}$. Consider a Cauchy sequence
$\{A\}_{n \in \bN} \subset \gB_1(\sH) \subset  \gB_1(\sH_\bR)$. Since $||A||_1 = \frac{1}{4}||A||^{\bR}_1$ from (\ref{4traces}) and (c), the sequence is also Cauchy in $\gB_1(\sH_\bR)$ and thus there is $A \in \gB_1(\sH_\bR)$ such that $A_n \to A$ in the norm
$||\:||_1^{\bR}$. However, from (e)(iii) whose proof does not depend on this argument, the convergence of the sequence is also valid referring to the operator norm $||\:||$ which is the same for the real and the quaternionic Hilbert space. In particular is also holds point-wise. Since every $A_n$
in $\gB(\sH_\bR)$ commutes with the three $\bR$-linear maps $\sH_\bR \ni x \mapsto xi\in \sH_\bR$, $\sH_\bR \ni x \mapsto xj\in \sH_\bR$ and $\sH_\bR \ni x \mapsto xk \in\sH_\bR$, also $A$ does and thus $A \in \gB_1(\sH)$ for {\bf (\ref{lemma}a)} and the last assertion in the proof of (a) above. This concludes the proof of (d). 

(e) (For $\sH$ quaternionic.)   The proof of (i) and (ii) immediately arises with the same argument exploited proving (d), using again $||A||_1 = \frac{1}{4}||A||^{\bR}_1$, the fact that (i) and (ii) are valid in the real Hilbert space case and the invariance of the norm $||A||$ when passing from $\gB(\sH)$ to $\gB(\sH_\bR)$. Let us conclude the proof by establishing the validity of (iii). Take $A\in \gB_1(\sH)$. Since $||Ax||^2 = \langle x|A^*Ax\rangle = || \:|A|x||^2$ so that $||A||= |||A|||$, proving the thesis amounts to prove that $|| \:|A|\: || \leq ||A||_1$. Since $|A|$ is compact  selfadjoint there is a Hilbert basis $M$ of $\sH$ made of eigenvectors $u$ with eigenvalues $s(u)\geq 0$
(since $|A|\geq 0$) according to Proposition \ref{propdecT} (it is sufficient to complete the basis $N$ introduced there). If $x\in \sH$,
$x = \sum_{u\in M} u x_u$ for quaternions $x_u \in \bH$. Therefore
$$||A||^2 = ||\:|A|\:||^2= \sup_{||x||=1}|| \:|A|\: x||^2 = \sup_{||x||=1}\left\langle \sum_{v\in M} v x_v  \left||A|^2 \sum_{u\in M} u x_u \right.\right\rangle =
\sup_{||x||=1} \sum_{u\in M} |x_u|^2 s(u)^2\:.$$
Since $1=||x||^2= \sum_{u\in M} |x_u|^2$, it must hold $|x_u|\leq 1$ and we can write
$$||A||^2\leq \sup_{||x||=1} \sum_u  s(u)^2=   \sum_u  s(u)^2 \leq \left(\sum_u  s(u)\right)^2 = ||A||^2_1\:,$$
where we  exploited $s(u) \geq 0$ and we made use of (c) in evaluating $||A||_1$ with respect to the said basis $M$ of eigenvectors of $|A|$.
\end{proof}

\begin{proof}[{\bf Proof of Proposition \ref{propTRACE}}]
	 Consider a Hilbert basis $N\subset \sH$ and a finite
	 subset $F\subset N$. We have, for $A \in \gB_1(\sH)$
using its polar decomposition $A=U|A|$,
$$\sum_{u\in F} |\langle u|A u\rangle| = \sum_{u\in F} |\langle u|U \sqrt{|A|}\sqrt{|A|} u\rangle| = 
 \sum_{u\in F} |\langle  \sqrt{|A|} U^* u|\sqrt{|A|} u\rangle| \leq \sum_{u\in F} ||\sqrt{|A|} U^* u||\:||\sqrt{|A|} u|| 
$$ $$ \leq \sqrt{\sum_{u\in F} ||   \sqrt{|A|} U^* u||^2} 
 \sqrt{\sum_{u\in F} ||   \sqrt{|A|} u||^2} = \sqrt{\sum_{u\in F}  \langle u|U|A|U^* u\rangle}\sqrt{\sum_{u\in F}  \langle u||A| u\rangle}\:.$$
Notice that $U|A|U^*$ is selfadjoint and positive, furthermore it belongs to $\gB_1(\sH)$ in view of (d) Proposition \ref{propTRACE1}. In summary, we have found that
$$\sum_{u\in F} |\langle u|A u\rangle| \leq \sqrt{\sum_{u\in F}  \langle u|U|A|U^* u\rangle}\sqrt{\sum_{u\in F}  \langle u||A| u\rangle} \leq \sqrt{|| U|A|U^*||_1}\sqrt{ ||A||_1}\:. $$
Taking advantage of (e) Proposition \ref{propTRACE1}, noticing that $||U||=||U^*||\leq 1$ since $U$ is a partial isometry, we have
$$\sum_{u\in F} |\langle u|A u\rangle| \leq  ||A||_1<+\infty\:.$$
Since this result holds true for every finite subset $F\subset N$, we conclude that 
$$\sum_{u\in N} |\langle u|A u\rangle| := \sup\left\{\left.\sum_{u\in F} |\langle u|A u\rangle| \:\right|\: \mbox{$F$ finite $\subset N$}\right\}\leq  ||A||_1<+\infty\:,$$
Thus only an countably at most number of elements  $ |\langle u|A u\rangle|$ do not vanish and the sum can computed as a standard series. As a consequence, the series  $\sum_{u\in N} \langle u|A u\rangle$  in $\bD$ absolutely converges and therefore can be re-ordered arbitarily without changing its sum.

(a) (i), (ii) and (iii) are easy consequences of the given definitions.

(b)  (i) and (ii)  are standard results. E.g., see Proposition 4.38 \cite{M} for the complex case and the real case can be identically proved. 

(c) Suppose that (i) holds. Define $N'$ out of $N$ just replacing  the element $u_0\in N$ for $u_0q \in N'$ for a given $q\in \bH$ with $|q|=1$ and keeping all remaining elements.
By direct computation $0= tr_{N'}(A)-tr_N(A) = \overline{q} \langle u_0|Au_0\rangle q -\langle u_0|Au_0\rangle$. Namely,
$\langle u_0|Au_0\rangle q = q\langle u_0|Au_0\rangle$.
Since $q$ is arbitrary, we have $\langle u_0|Au_0\rangle \in \bR$. Every unit vector $u_0 \in \sH$ can be completed to a Hilbert basis and therefore  $\langle x|Ax\rangle = \overline{\langle x|Ax\rangle} = \langle Ax|x\rangle$ for $x\in \sH$. In summary $\langle x| (A-A^*)x\rangle =0$ for every $x\in \sH$. By polarization $A=A^*$. Suppose conversely that $A^*=A \in \gB_1(\sH)$. Since $A \in \gB_\infty(\sH)$ for (b) Proposition \ref{propTRACE1}, we can decompose $A$ along a Hilbert basis of eigenvectors $M$ (completing the Hilbert basis $N \subset Ker(A)^\perp$ in Proposition \ref{propdecT}),
\begin{equation}
A = \sum_{u\in M} u s(u) \langle u| \:\: \rangle  \label{decAM}\:,
\end{equation}
 and in particular using $M$ to compute the trace, with $s(u)=0$ if $u\in ker(A)$:
$$
tr_M(A) = \sum_{u\in M}\langle u|A u\rangle = \sum_{u\in M} s(u)\:.
$$
As shown in the first part of this very proof this sum is absolutely convergent. In particular, since $s(u)=\langle u| A u\rangle$ for every $u\in M$, it holds that
$$
%||A||_1=tr_M(|A|) = 
\sum_{u\in M} |s(u)| < +\infty\:.
$$
Next consider another Hilbert basis $B\in \sH$, and from (\ref{decAM}) we easily find
$$\langle v| Av \rangle = \sum_{u\in M} s(u) |\langle v|u \rangle|^2 \quad \forall v \in B\:,$$
so that
\begin{equation}\label{swap}
\sum_{v\in B}\langle v| Av \rangle = \sum_{v\in B}\sum_{u\in M} s(u) |\langle v|u \rangle|^2 \quad \forall v \in B\:.
\end{equation}
If we were allowed to swap the two summations in the right-hand side of (\ref{swap}), we would obtain
\beq tr_B(A) =\sum_{v\in B}\langle v| Av \rangle =\sum_{u\in M}  s(u) \sum_{v\in B} |\langle v|u \rangle|^2 = \sum_{u\in M}  s(u) ||u||^2 = \sum_{u\in M}  s(u) = tr_M(A)\:,\label{END}\eeq
proving that the trace does not depend on the used Hilbert basis and concluding the proof.
To prove that those summations can be in fact interchanged, first observe that, in double summation in  (\ref{swap}), $u$
varies in a set  at most countable according to  Proposition \ref{propdecT}): Only the elements with eigenvalue $s(u)\neq 0$ give contribution and every eigenspace but the kernel of $A$ have finite dimension. These $u$ are exactly the elmentents of the countable Hilbert basis $N$ of $Ker(A)^\perp$.  Similarly, only a set at most countable of elements $v$ gives contribution to the second sum (for every $u$, the set of non-vanishing Fourier coefficients along $B$ is at most countable). Summing up, the double sum in (\ref{swap}) can be intepreted as an iterated integral with respect the product measure over  $N \times B'$ of a pair of $\sigma$-finite counting measures.
$B'$ is countable as it is defined as the union of the supports (including an at most countable set of elements) of the functions $f_u : B \ni v \mapsto |\langle v|u \rangle|^2$ for  $u$ varying in the at most countable set $N$.
 The integrals can be interchanged (Fubini) provided the function
$N \times B' \ni (u,v) \mapsto s(u) |\langle v|u \rangle|^2 $
 is absolutely integrable with respect the product measure. In turn, this is equivalent (Tonelli theorem) to requiring that one of the two iterated integrals, either 
$$
\sum_{u\in N}\sum_{v\in B'} |s(u)| |\langle v|u \rangle|^2< +\infty \quad \mbox{or}\quad \quad \sum_{v\in B'} \sum_{u\in N} |s(u)| |\langle v|u \rangle|^2<+\infty\:.
$$
 This is the case, indeed, $M\times B \supset N\times B'$ and
$$
\sum_{u\in M}\sum_{v\in B} |s(u)| |\langle v|u \rangle|^2 =\sum_{u\in M} |s(u)| \sum_{v\in B}  |\langle v|u \rangle|^2 = \sum_{u\in M} |s(u)|  <+\infty\:.
$$
Therefore (\ref{END}) is valid, concluding the proof of  (c).

(d) We have to consider the case $\bD=\bH$, since the validity of (i) and (ii) for $\bD=\bR, \bC$ immediately follows from (b).

 (i) (For $\bD=\bH$) It arises from (a) and (c):  $2Re(tr_N(A)) = tr_N(A)+ \overline{tr_N(A)} = tr_N(A)+ tr_N(A^*) = tr_N(A+A^*)$ does not depend on $N$ since $A+A^*$ is selfadjoint. 
 
(ii) (For $\bD=\bH$) We consider the real Hilbert space $\sH_\bR$ as in Lemma \ref{lemma}. 
According to  (\ref{lemma}b) in  Lemma \ref{lemma}, if $N$ is a Hilbert basis of $\sH$, 
$N_\bR$ is a Hilbert basis of the real Hilbert space $\sH_\bR$ equipped with the real scalar product $(\:\:|\:\:) = Re(\langle \:\:|\:\:\rangle)$.
In view of the identity (\ref{4traces}), if $AB \in \gB_1(\sH)$, then  $AB \in \gB_1(\sH_\bR)$ and the same fact holds true for $BA$, so we can compute the trace $tr_{N_\bR}(AB)$ and $tr_{N_\bR}(BA)$  in the Hilbert space $\sH_\bR$ where the cyclic property is valid for (b). Using the real scalar product $(\:\:|\:\:) = Re(\langle \:\:|\:\:\rangle)$ of $\sH_\bR$ and continuity of the function $Re$, we therefore have
$$\sum_{z\in N} Re(\langle z|ABz \rangle)+ \sum_{z\in N} Re(\langle zi|AB zi \rangle)+ \sum_{z\in N} Re(\langle zj|AB zj \rangle)+\sum_{z\in N} Re(\langle zk|AB zk \rangle)$$
$$ = tr_{N_\bR}(AB) = tr_{N_\bR}(BA) =$$
$$=\sum_{z\in N} Re(\langle z|BA z \rangle)+ \sum_{z\in N} Re(\langle zi|BA zi \rangle)+ \sum_{z\in N} Re(\langle zj|BA zj \rangle)+\sum_{z\in N} Re(\langle zk|BA zk \rangle)\:.$$
Next observe that, since $Re(qq')= Re(q'q)$, it holds that $Re(\langle zi|T zi \rangle) = -Re(i(\langle z|T z \rangle i)) =
-Re(i\:  i\langle z|T z \rangle) = Re(\langle z|T z \rangle)$. Similarly 
$Re(\langle jz|T jz \rangle)= Re(\langle kz|T kz \rangle) = Re(\langle z|T z \rangle)$ and therefore, the long identity written above can be rephrased into
$$4\sum_{z\in N} Re(\langle z|ABz \rangle) = tr_{N_\bR}(AB) = tr_{N_\bR}(BA) = 4\sum_{z\in N} Re(\langle z|BA z \rangle)\:.$$
In summary, $Re(tr_N(AB))= Re(tr_N(BA))$.

Let us prove point (e). 
Let $N$ a Hilbert basis as in the hypotheses obtained by completing a Hilbert basis $N_e$ of $Ker(A)^\perp$ made of eigenvectors of $A$. Then we have
\begin{equation}
\begin{split}
tr_N(BA)&=\sum_{u\in N}\langle u|BAu\rangle =\sum_{u\in N_e}\langle u|Bu\rangle s(u)=\sum_{u\in N_e}\langle Au|Bu\rangle =\sum_{u\in N}\langle Au|Bu\rangle=\\
&=\sum_{u\in N}\langle u|ABu\rangle=tr_N(AB)\:,
\end{split}
\end{equation}
where $s(u)\in\bR$ is the eigenvalue associated with the eigenvector $u\in N_e$. In particular, if $B=B^*$, then $\overline{tr_N(AB)}=tr_N((AB)^*)=tr_N(BA)=tr_N(AB)$, and so $tr_N(AB)\in\bR$.

To conclude let us prove point (f).  The condition $A\ge B$ means that $\langle x|Ax \rangle - \langle x|Bx\rangle \ge 0$, thus in particular $Re \langle x|Ax\rangle \ge Re\langle x|Bx \rangle$, for all $x\in\sH$. So, consider any Hilbert basis $N\subset \sH$, then
$$
Re(tr_N(A))=\sum_{z\in N}Re \langle z|Az \rangle \ge \sum_{z\in N}Re \langle z|Bz \rangle =Re(tr_N(B))\:.
$$
\end{proof}

\noindent{\bf Proof of Proposition \ref{proptraceeq}}. (i) implies (ii) as established with the first statement of Proposition \ref{propTRACE}. The fact that (ii) entails (i) if $\bD=\bC$ is a known result (e.g., see Proposition 4.41 in \cite{M}). Failure of the implication (ii) $\Rightarrow$ (i) in the real case is evident considering $A: \sH \to \sH$ with $AA=-I$ and $A^*=A$ (such an  operator can be constructed easily referring to a Hilbert basis of an infinite dimensional real Hilbert space) so that $|A|=I$ and $\sum_{u\in N} \langle u||A|u\rangle = +\infty$ falsifying (i), but 
$\langle u|Au\rangle = -\langle Au|u\rangle = -\langle u|Au\rangle$ leads to 
$\sum_{u\in N} \langle u|A u\rangle = 0$, satisfying (ii).
 To conclude, we prove that  (ii) $\Rightarrow$ (i) for $\bD=\bH$. Suppose that $A\in \gB(\sH)$ satsfies (ii). Define $B := (A+ A^*)/2$ and $C= (A-A^*)/2$. Since $|\langle x|(A\pm A^*) x\rangle| = |\langle x|A x\rangle \pm  \overline{\langle x|A x\rangle}| \leq 2|\langle x|A x\rangle|$, both $B$ and $C$ satisfy (ii). Since $\gB_1(\sH)$ is a real vector space ((d) Proposition \ref{propTRACE1}), it is sufficient to prove that $B$ and $C$ satisfy (i)  to conclude. Regarding the selfadjoint operator $B$, the fact that it satisfies (i) from (ii) can be proved passing to its integral spectral decomposition (using the fact that the spherical spectrum is completely included in $\bR$ since $B=B^*$) and following the same route as that used to prove Proposition 4.41 in \cite{M} for $T=T^*$. Let us pass to $C$. Due to Theorem 5.9 in \cite{GMP1} (specialised to $T=C$), we can write  $C=J|C|$ for some $J \in \gB(\sH)$ with $JJ = -I$ and $J^*=-J$
such that $JC=CJ$ and $J|C|=|C|J$. As a consequence  (\cite{GMP1} Sect.3.3), the complex Hilbert space $\sH_{J\imath} := \{u \in \sH \:|\: Ju =u\imath\}$ equipped with the restriction of the scalar product of $\sH$ is invariant under $C$. If $N\subset \sH_{J\imath} $ is a Hilbert basis, it is also a Hilbert basis of $\sH$ in view of (f) Proposition 3.8 in \cite{GMP1}. So,
condition (ii) for $C$ specializes to  $\sum_{u\in N} |\langle u|J|C|u\rangle| < +\infty$ for the said simultaneous Hilbert  basis of $\sH_{J\imath}$ and $\sH$. From the definition of   $\sH_{J\imath}$,
$\sum_{u\in N} |\imath\langle u||C|u\rangle| < +\infty$, which can be written as $\sum_{u\in N} \langle u||C|u\rangle < +\infty$ so that $C\in \gB_1(\sH)$ by definition, concluding the proof.  $\Box$\\

\noindent{\bf Proof of Proposition \ref{lastPROP}}.
Consider a Hilbert basis $N$ of $\sH$ such that $N \subset \sH_{J\imath}$. We have
$$tr_N(A)= \frac{1}{2}tr_N(A+A^*) +  \frac{1}{2}tr_N(A-A^*) = Re(tr_N(A)) +  \frac{1}{2}tr_N(J|A-A^*|)\:.$$
We have proved that
$$tr_N(A) = tr^\bR(A)+   \frac{1}{2}\sum_{z\in N}\langle z|J|A-A^*|z\rangle
= tr^\bR(A)-  \frac{1}{2}\sum_{z\in N}\langle Jz||A-A^*|z\rangle \:.$$
Since $Jz=z\imath$ and $|A-A^*|$ is selfadjoint so that $\langle z||A-A^*|z\rangle  \in \bR$,
$$tr_N(A) = tr^\bR(A)+  \frac{\imath}{2}\sum_{z\in N}\langle z||A-A^*|z\rangle 
=tr^\bR(A)+  \frac{\imath}{2}tr^\bR(|A-A^*|) = tr^\bR(A)+  \frac{\imath}{2}tr(|A-A^*|) \:,$$
where we have used (e) Corollary \ref{corollary}.
Eventually, observe that any change of $J$ on $Ker(A-A^*)= Ker(|A-A^*|)$ does not affect the result because it would only change vanishing  terms $\langle z|J|A-A^*|z\rangle$ in view of $|A-A^*|z=0$. $\Box$\\

\noindent{\bf Proof of Lemma \ref{lemmatecnico}}.
Let $N$ be as in the hypotheses, and consider the decomposition
	$$
	1=\sum_{n\in N}p_nq_n= (1-p_0)\left[\sum_{n=1}^N\frac{p_n}{1-p_0}q_n\right]+p_0q_0\:.
	$$
	Notice the following inequalities
	\begin{equation}
	0<\frac{p_i}{1-p_0}< \sum_{n=1}^N\frac{p_n}{1-p_0}=\frac{1}{1-p_0}\sum_{n=1}^Np_n=\frac{1}{1-p_0}(1-p_0)=1\:,
	\end{equation}
	from which we also have
	\begin{equation}
	0\le q:=\sum_{n=1}^N\frac{p_n}{1-p_0}q_n\le \sum_{n=1}^N\frac{p_n}{1-p_0}=1\:.
	\end{equation}
	Thus we reduce to $(1-p_0)q+p_0q_0=1$ with $p_0\in (0,1)$ and $q_0,q\in [0,1]$. Assume by contradiction that $q\neq q_0$ and, without loss of generality, suppose that $q_0>q$. The identity above can be rewritten as $q-p_0q+p_0q_0=1$, i.e. $p_0(q_0-q)=1-q$. Thus, from 
	$
	(1-q)(q_0-q)^{-1}=p_0 < 1
	$
	we get $1-q<q_0-q$ which is equivalent to $1<q_0$, in turn impossible. Thus $q=q_0$ and more precisely $1=(1-p_0)q+p_0q_0=q_0=q$. 
	Repeating the argument on the sums $\sum_{n=1}^N\frac{p_n}{1-p_0}=\sum_{n=1}^N\frac{p_n}{1-p_0}q_n=1$ we get $q_1=1$. By induction we get $q_n=1$ for all $n$. $\Box$


\begin{thebibliography}{999}

\bibitem[ACK16]{ACK16} D Alpay, F Colombo, D.P. Kimsey: {\em The spectral theorem for quaternionic unbounded normal operators based on the S-spectrum}.  J. Math. Phys. {\bf 57}, 023503 (2016)

\bibitem[Ad95]{Adler} S. L. Adler: {\em Quaternionic Quantum Mechanics and Quantum Fields}. International
Series of Monographs on Physics, Vol. 88 The Clarendon Press Oxford University
Press, New York, (1995)

\bibitem[AeSt00]{XY} D. Aerts, B. van Steirteghem: {\em Quantum Axiomatics and a theorem of M.P. Sol\`er}.  Int. J. Theor. Phys. {\bf 39}, 497-502, (2000).


\bibitem[BeCa81]{BeCa}E.G., Beltrametti, G. Cassinelli: {\em The logic of quantum mechanics}. Encyclopedia of Mathematics and its Applications, vol. 15, Addison-Wesley, Reading, Mass., (1981)


\bibitem[BiNe36]{BiNe}
G.  Birkhoff  and  J.  von  Neumann.   {\em The  logic  of  quantum
mechanics.}
Annals of Mathematics, {\bf 37} 823-843 (1936).

\bibitem[Bre83]{Brezis} H. Brezis. {\em Analyse fonctionnelle; th\'eorie et applications}, Masson (1983)

\bibitem[CGSS07]{7}  F. Colombo, G. Gentili, I. Sabadini and D. C. Struppa, {\em A functional calculus in a
noncommutative setting}, Electron. Res. Announc. Math. Sci. 14 (2007) 60-68.

\bibitem[CSS11]{8} F. Colombo, I. Sabadini and D. C. Struppa, {\em Noncommutative Functional Calculus}, Progress in Mathematics, Vol. 289 (Birkh\"auser/Springer Basel AG, Basel, 2011), Theory and applications of slice hyperholomorphic
functions;


\bibitem[CGJ16]{CGJ16} F. Colombo, J. Gantner, T. Janssens: {\em Schatten class and Berezin transform of quaternionic linear operators}. Mathematical Methods in the Applied Sciences, {\bf 39}, 5582-5606 (2016).



\bibitem[Dv93]{libroGleason} A. Dvurecenskij: {\em  Gleason's Theorem and Its Applications}, Kluwer Academic Publishers (1993)


\bibitem[EGL09]{librone} K. Engesser, D.M. Gabbay, D. Lehmann  (editors): {\em Handbook of Quantum Logic and Quantum Structures}. Elsevier, Amsterdam (2009) 



\bibitem[FJSS62]{foudationofquaternionicmechanics} D. Finkelstein, J. M. Jauch, S. Schiminovich, and D. Speiser: {\em Foundations of Quaternion Quantum Mechanics}. J. Math. Phys. {\bf 3}, 207 (1962)

\bibitem[Gan17]{Gan17} J. Gantner, {\em On the equivalence of complex and quaternionic quantum mechanics},  arXiv:1709.07289


\bibitem[Gl57]{G}  A.M. Gleason, {\em Measures on the closed subspaces of a Hilbert space}. J. Math. Mech. {\bf 6}(6), 885-893 (1957)

\bibitem[GMP13]{GMP1} R. Ghiloni, V. Moretti  and A. Perotti: {\em Continuous slice functional calculus in quaternionic Hilbert spaces}
Rev. Math. Phys. {\bf 25}, (2013) 1350006, 


\bibitem[GMP14]{GMP3} R. Ghiloni, V. Moretti  and A. Perotti: {\em Spectral properties of compact normal quaternionic operators}, in {\em Hypercomplex Analysis: New perspectives and applications}, (Eds S. Bernstein, U. Kaehler, I. Sabadini, F. Sommen),
Trends in Mathematics, Birkhauser, Basel (2014) 



\bibitem[GMP17]{GMP2} R. Ghiloni, V. Moretti  and A. Perotti: {\em Spectral representations of normal operators via Intertwining Quaternionic Projection Valued Measures}
 Rev. Math. Phys. {\bf 29}, 1750034 (2017)


\bibitem[Lan17]{Landsman} K. Landsman: Fo{\em ndations of Quantum Theory},  Springer, (2017)

\bibitem[Li03]{Li} B. Li: {\em Real Operator Algebras}. World Scientific, (2003)


\bibitem[Ma63]{Mackey} G. Mackey: {\em The Mathematical Foundations of Quantum Mechanics}. Benjamin, New York (1963)

\bibitem[Mo18]{M} V. Moretti: {\em Spectral Theory and Quantum Mechanics, Mathematical Structure of
Quantum Theories, Symmetries and introduction to the Algebraic Formulation}. 2nd Edition, Springer, 2018
 
\bibitem[MoOp17]{MO1} V. Moretti and M. Oppio: {\em  Quantum theory in real Hilbert space: How the complex Hilbert space structure emerges from Poincar\'e symmetry}.
Rev. Math. Phys. 29  (2017)  1750021     

\bibitem[MoOp17b]{MO2} V. Moretti  and M. Oppio: {\em Quantum theory in quaternionic Hilbert space: How Poincar\'e symmetry reduces the theory to the standard complex one}.
arXiv:1709.09246  




\bibitem[Ru91]{R}  W. Rudin, {\em Functional Analysis} 2nd edition, Mc Graw Hill, (1991)



\bibitem[Sc12]{S} K. Schm\"udgen, {\em Unbounded Self-adjoint Operators on Hilbert Space}, Springer,  2012


\bibitem[So95]{Soler} M.P. Sol\`er:  {\em Characterization of Hilbert spaces by orthomodular spaces}. Communications in Algebra, {\bf 23}, 219-243 (1995)


\bibitem[St60]{S1}   E.C.G. St\"uckelberg: {\em Quantum Theory in Real Hilbert Space}. Helv. Phys. Acta,
{\bf 33}, 727-752, (1960)

\bibitem[StGu61]{S2} E.C.G. St\"uckelberg and M. Guenin: {\em Quantum Theory in Real Hilbert Space II} (Addenda and Errats). Helv. Phys. Acta, {\bf 34} :621-628, (1961).

\bibitem[Tor95]{T} A. Torgasev: {\em Quaternionic Operators with Finite Matrix Trace}. Integr Equat Oper Th {\bf 23}, 114, (1995)

\bibitem[Va68]{V2a}  V.S. Varadarajan, {\em Geometry of Quantum Theory}.  Van Nostrand Reinhold Inc. (1968)


\bibitem[Va07]{V2}  V.S. Varadarajan, {\em Geometry of Quantum Theory}. 2nd Edition, Springer (2007)





\end{thebibliography}
\end{document}